\documentclass[11pt,a4paper,twoside]{article}
\usepackage{titling}
\usepackage{titletoc}
\usepackage{url}

\usepackage[utf8]{inputenc}
\usepackage[T1]{fontenc}
\usepackage{amsmath}
\usepackage{amsfonts}
\usepackage{amssymb}
\usepackage{amsthm}

\usepackage{mathrsfs}
\usepackage{mathtools}
\usepackage{stmaryrd}
\usepackage[all]{xy}
\usepackage{makeidx}
\usepackage{natbib}
\usepackage[french,english]{babel}
\usepackage{xspace}
\usepackage{float}

\usepackage{etoolbox}

\makeatletter\def\th@plain{\slshape}\makeatother
\makeatletter\patchcmd{\th@remark}{\itshape}{\slshape}{}{}\makeatother

\newcounter{bidon}
\newcommand{\rdb}{\refstepcounter{bidon}}

\renewcommand\paragraph[1]{
\rdb
\addcontentsline{toc}{paragraph}{#1} 
\medskip \noindent $\bullet$ \textbf{#1}}

\newcommand {\junk}[1]{}

\floatstyle{boxed}
\floatname{agc}{Circuit}
\newfloat{agc}{ht}{lag}[section]
\floatname{agC}{Circuit}
\newfloat{agC}{H}{lag}[section]

\newenvironment{algor}[1][]
{\par\smallskip\begin{agc} \vskip 1mm
\begin{algorithm}{\bfseries#1}
\upshape\sffamily
}
{\end{algorithm}
\end{agc}
}

\newenvironment{algorH}[1][]%
{\par\smallskip\begin{agC}\vskip 1mm
\begin{algorithm}{\bfseries#1}
\upshape\sffamily
}
{\end{algorithm}
\end{agC}
}

\newenvironment{falgor}[1][]
{\par\smallskip\begin{agc} \vskip 1mm
\begin{falgorithm}{\bfseries#1}
\upshape\sffamily
}
{\end{falgorithm}
\end{agc}
}

\newenvironment{falgorH}[1][]%
{\par\smallskip\begin{agC}\vskip 1mm
\begin{falgorithm}{\bfseries#1}
\upshape\sffamily
}
{\end{falgorithm}
\end{agC}
}

\newcommand\hsu{\\ \hspace*{4mm}}
\newcommand\hsd{\\ \hspace*{8mm}}
\newcommand\hst{\\ \hspace*{1,2cm}}

\newcommand \noi {\noindent}

\newcommand \ms {\medskip}
\newcommand \mni {\ms\noi}
\newcommand \bs {\bigskip}
\newcommand \bni {\bs\noi}

\newcommand\gen[1]{\left\langle{#1}\right\rangle}

\newcommand \s[1] {\underline#1 }

\newcommand\adots{\mathinner{\mkern0mu\raise 1pt\hbox{\string.}\mkern
3mu\raise 4pt\hbox{\string.}\mkern 3mu\raise 7pt\hbox{\string.}}}

\newcommand \K{\mathbb{K}}

\newcommand \N{\mathbb{N}}

\newcommand \gk{{\bf k}}
\newcommand \gA{\mathbf{A}}
\newcommand \gB{\mathbf{B}}
\newcommand \gC{\mathbf{C}}
\newcommand \gK{\mathbf{K}}
\newcommand \gL{\mathbf{L}}

\newcommand \rP {{\rm P}}
\newcommand \rH {{\rm H}}

\newcommand \Ker {{\rm Ker}}
\renewcommand \mod {{\rm \;mod\;}}
\newcommand \cd {{\rm lc}}

\usepackage[bookmarksopen=false,breaklinks=true,%
      backref=page,pagebackref=true,plainpages=false,%
      hyperindex=true,pdfstartview=FitH,%
      pdfpagelabels=true,
      colorlinks=true,linkcolor=blue,%
      citecolor=red,urlcolor=red,
      colorlinks=true%
      ]%
   {hyperref}

\begin{document}

\thispagestyle{empty}
~ 
\vspace{3cm}

\noindent In this file you find the English version starting on the page  numbered \pageref{beginenglish}.

\medskip \noindent  {\Large \bf The Berlekamp-Massey Algorithm revisited}

\bigskip \noindent  
Then the French version begins on the page numbered \pageref{beginfrench}.

\medskip\noindent   {\Large \bf L'algorithme de Berlekamp-Massey revisité} 

\bigskip  The English version of this paper appeared in AAECC {\bf 17} 1 (2006), 75--82.
Received: 14 September 2005 / Published online: 9 March 2006.

\medskip \noindent {\bf Note.} The paper was submitted to \textsl{Mathematics of Computations} in 2004. The referee wrote that our algorithm was well known but did not give any reference in the literature.

\newpage
\thispagestyle{empty}

~

\pagestyle{headings}
\patchcmd{\sectionmark}{\MakeUppercase}{}{}{}
\setcounter{page}{0}\renewcommand\thepage{E\arabic{page}}

\pagestyle{headings}
\patchcmd{\sectionmark}{\MakeUppercase}{}{}{}

\def\proofname{\textsl{Proof}}

\begingroup

\title{ The Berlekamp-Massey Algorithm revisited}

\author{
Nadia Ben Atti (\thanks {~ Equipe de Math\'ematiques, CNRS UMR
6623, UFR des Sciences et Techniques, Universit\'e de
Franche-Comt\'e, 25 030 Besan\c{c}on cedex, France. {\tt
nadia.benatti@ensi.rnu.tn}}~), Gema M. Diaz--Toca (\thanks {~
Dpto. de Matematicas Aplicada. Universidad de Murcia, Spain. {\tt
gemadiaz@um.es}, partially supported by the Galois Theory and
Explicit Methods in Arithmetic Project HPRN-CT-2000-00114}~) Henri
Lombardi (\thanks{~ Equipe de Math\'ematiques, CNRS UMR 6623, UFR
des Sciences et Techniques, Universit\'e de Franche-Comt\'e, 25
030 Besan\c{c}on cedex, France, {\tt
henri.lombardi@univ-fcomte.fr}, partially supported by the European
Union funded project RAAG CT-2001-00271}~) }

\date{2005}
\maketitle

\startcontents[english]
\setcounter{tocdepth}{4}

\newcommand\sialors [1]{\textbf{si } $#1$ \textbf{ then }}

\newcommand\pour[3]{\textbf{for } $#1$ \textbf{ from } $#2$
      \textbf{ to} $#3$ \textbf{do  }}
\newcommand\tantque[1]{\textbf{while} $#1$ \textbf{ do }}
\newcommand\finpour{\textbf{end for}}
\newcommand\fintantque{\textbf{end while }}
\newcommand\Debut{\\[1mm] \textbf{Start }}
\newcommand\fin{\textbf{\\ End.}}
\newcommand\Entree{\\[1mm] \textbf{Input: }}
\newcommand\Sortie{\\ \textbf{Output : }}
\newcommand\Varloc{\\ \textbf{Local variables : }}

\newcommand \srl{ linearly recurrent sequence\xspace}

\newcommand\comm{\rdb
\noi{\it Comment. }}

\newcommand\COM[1]{\rdb
\noi{\it Comment #1. }}

\newcommand\comms{\rdb
\noi{\it Comments. }}

\newcommand\Pb{\rdb
\noi{\bf Problem. }}

\newcommand \rem{\rdb
\noi{\sl Remark. }}

\newcommand \REM[1]{\rdb
\noi{\sl Remark#1. }}

\newcommand \rems{\rdb
\noi{\sl Remarks. }}

\newcommand \exl{\rdb
\noi{\bf Example. }}

\newcommand \EXL[1]{\rdb
\noi{\bf Example: #1. }}

\newcommand \exls{\rdb
\noi{\bf Examples. }}

\newcommand\gui[1]{``{#1}''}

\newcommand \thref[1] {Theorem~\ref{#1}}
\newcommand \paref[1] {page~\pageref{#1}}
\newcommand \pstfref[1] {Positivstellensatz formel~\ref{#1}}
\newcommand \pstref[1] {Positivstellensatz~\ref{#1}}

\newcommand\subsubsec[1] {\subsubsection*{#1}}

\newcommand \hdr {induction hypothesis\xspace}
\newcommand \ssi {if and only if\xspace}
\newcommand \cnes {necessary and sufficient condition\xspace}
\newcommand \spdg {without loss of generality\xspace}
\newcommand \Propeq {T.F.A.E.\xspace}
\newcommand \propeq {t.f.a.e.\xspace}
\newcommand \disept {17$^{th}$ Hilbert's problem\xspace}

\def \cad {\textit{i.e.}\ }


\newcommand \Amo {$\gA$-module\xspace}
\newcommand \Amos {$\gA$-modules\xspace}

\newcommand \Bmo {$\gB$-module\xspace}
\newcommand \Bmos {$\gB$-modules\xspace}

\newcommand \Zmo {$\gZ$-module\xspace}
\newcommand \Zmos {$\gZ$-modules\xspace}

\newcommand \ZZmo {$\ZZ$-module\xspace}
\newcommand \ZZmos {$\ZZ$-modules\xspace}

\newcommand \Ali {$\gA$-\ali}
\newcommand \Alis {$\gA$-\alis}

\newcommand \Alg {$\gA$-\alg}
\newcommand \Algs {$\gA$-\algs}

\newcommand \kev {$\gk$-vector space\xspace}
\newcommand \kevs {$\gk$-vector spaces\xspace}

\newcommand \Kev {$\gK$-vector space\xspace}
\newcommand \Kevs {$\gK$-vector spaces\xspace}

\newcommand \klg {$\gk$-\alg}
\newcommand \klgs {$\gk$-\algs}

\newcommand \Klg {$\gK$-\alg}
\newcommand \Klgs {$\gK$-\algs}


\newcommand \agq {algebraic\xspace}

\newcommand \alg {algebra\xspace}
\newcommand \algs {algebras\xspace}

\newcommand \agB {Boolean \alg}

\newcommand \algo{algorithm\xspace}
\newcommand \algos{algorithms\xspace}

\newcommand \algq{algorithmic\xspace}

\newcommand \ali {\lin map\xspace}
\newcommand \alis {\lin maps\xspace}

\newcommand \anar {\ari \ri}
\newcommand \anars {\ari \ris}
\newcommand \Anars {\Ari \ris}

\newcommand \ari{arith\-metic\xspace}

\newcommand \auto {automorphism\xspace}
\newcommand \autos {automorphisms\xspace}


\newcommand \cac {algebraically closed field\xspace}
\newcommand \cacs {algebraically closed fields\xspace}

\newcommand \carn{characterization\xspace}  
\newcommand \carns{characterizations\xspace}  

\newcommand \coe {coefficient\xspace}
\newcommand \coes {coefficients\xspace}

\newcommand \coh {coherent\xspace}
\newcommand \cohz {coherent\xspace}

\newcommand \coli {linear combination\xspace}
\newcommand \colis {linear combinations\xspace}

\newcommand \com {comaximal\xspace}

\newcommand \coo {coordinate\xspace}
\newcommand \coos {coordinates\xspace}


\newcommand \ddp {Pr\"ufer domain\xspace}
\newcommand \ddps {Pr\"ufer domains\xspace}

\newcommand \ddk {Krull dimension\xspace}

\newcommand \dfn{definition\xspace}  
\newcommand \dfns{definitions\xspace}  

\newcommand \discri{discriminant\xspace}
\newcommand \discris{discriminants\xspace}

\newcommand \dok {Dedekind domain\xspace}
\newcommand \doks {Dedekind domains\xspace}

\newcommand \dve {divisibility\xspace}

\newcommand \dvz {zero divisor\xspace}
\newcommand \dvzs {zero divisors\xspace}

\newcommand \eco{\com \elts}  

\newcommand \egt{equality\xspace} 
\newcommand \egts{equalities\xspace} 

\newcommand \elr{elementary\xspace}  

\newcommand \elt{element\xspace}  
\newcommand \elts{elements\xspace}  

\def \endo {endomorphism\xspace}
\def \endos {endomorphisms\xspace}

\newcommand \entrel {entailment relation\xspace}
\newcommand \entrels {entailment relations\xspace}

\newcommand \eqv  {equivalent\xspace}

\newcommand \evc{vector space\xspace} 
\newcommand \evcs{vector spaces\xspace} 


\newcommand \fab {bounded \fcn}
\newcommand \fabs {bounded \fcns}

\newcommand \fac {total \fcn}
\newcommand \facz {total \fcnz}

\newcommand \fap {partial \fcn}
\newcommand \faps {partial \fcns}

\newcommand \fcn {factorisation\xspace}
\newcommand \fcns {factorisations\xspace}

\newcommand \fdi{strongly discrete\xspace} 


\newcommand\gmq{geometric\xspace}

\newcommand\gne{generalized\xspace}

\newcommand\gnl{general\xspace}

\newcommand\gnlt{generally\xspace}

\newcommand\gnn{generalisation\xspace}
\newcommand\gnns{generalisations\xspace}

\newcommand\gnq{generic\xspace}

\newcommand\grl{$\ell$-group\xspace}
\newcommand\grls{$\ell$-groups\xspace}

\newcommand \gtr{generator\xspace}  
\newcommand \gtrs{generators\xspace}  


\newcommand \homo {homomorphism\xspace}
\newcommand \homos {homomorphisms\xspace}

\newcommand \id {ideal\xspace}
\newcommand \ids {ideals\xspace}

\newcommand \idd {de\-ter\-mi\-nantal \id}
\newcommand \idds {de\-ter\-mi\-nantal \ids}
\newcommand \iddz {de\-ter\-mi\-nantal \idz}
\newcommand \iddsz {de\-ter\-mi\-nantal \idsz}

\newcommand \idema {maximal \id}
\newcommand \idemas {maximal \ids}

\newcommand \idep {prime \id}
\newcommand \ideps {prime \ids}

\newcommand \idemi {minimal prime\xspace}
\newcommand \idemis {minimal primes\xspace}

\newcommand \idf {Fitting \id}
\newcommand \idfs {Fitting \ids}

\newcommand \idm {idempotent\xspace}
\newcommand \idms {idempotents\xspace}

\newcommand \idtr {indeterminate\xspace}
\newcommand \idtrs {indeterminates\xspace}

\newcommand \ifr {fractional \id}
\newcommand \ifrs {fractional \ids}

\newcommand \itf {\tf \id}
\newcommand \itfs {\tf \ids}

\newcommand \iso {isomorphism\xspace}
\newcommand \isos {isomorphisms\xspace}

\newcommand \iv {invertible\xspace}

\newcommand \lec {reader\xspace}

\newcommand \lgb {local global\xspace}

\newcommand \lin {linear\xspace}

\newcommand \lon {localisation\xspace}
\newcommand \lons {localisations\xspace}

\newcommand \lop {\lot principal\xspace}

\newcommand \losd {\lot \sdz\xspace}

\def \lot {locally\xspace}

\newcommand \mlp {principal \lon matrix\xspace}
\newcommand \mlps {principal \lon matrices\xspace}

\newcommand \mnp {manipulation\xspace}
\newcommand \mnps {manipulations\xspace}
\newcommand \mnr {\elr \mnp}
\newcommand \mnrs {\elr \mnps}

\newcommand \mo {monoid\xspace}
\newcommand \mos {monoids\xspace}
\newcommand \moco {\com \mos}

\newcommand \mpf {\pf module\xspace}
\newcommand \mpfs {\pf modules\xspace}

\newcommand \mpn {\pn matrix\xspace}
\newcommand \mpns {\pn matrices\xspace}

\newcommand \mpr {\pro module\xspace}
\newcommand \mprs {\pro modules\xspace}

\newcommand \mprn {\prn matrix\xspace}
\newcommand \mprns {\prn matrices\xspace}

\newcommand \mptf {\ptf module\xspace}
\newcommand \mptfs {\ptf modules\xspace}

\newcommand \mrc {projective module of constant rank\xspace}
\newcommand \mrcs {projective modules of constant rank\xspace}


\newcommand \ncr{necessary\xspace}

\newcommand \ncrt{necessarily\xspace}

\newcommand \ndz {regular\xspace}

\newcommand \noe {Noetherian\xspace}
\newcommand \noco {\noe\coh}

\newcommand \nst {Nullstellensatz\xspace}
\newcommand \nsts {Nullstellens\"atze\xspace}

\newcommand \odz {Zariski open set\xspace}

\newcommand \oqc {\qc open set\xspace}
\newcommand \oqcs {\qc open sets\xspace}


\newcommand \pa {saturated pair\xspace}
\newcommand \pas {saturated pairs\xspace}

\newcommand \pb{problem\xspace}  
\newcommand \pbs{problems\xspace}

\newcommand \peq {purely equational\xspace}

\newcommand \pf {finitely presented\xspace}

\newcommand \plg {\lgb principle\xspace}
\newcommand \plgs {\lgb principles\xspace}

\newcommand \pn {presentation\xspace}
\newcommand \pns {presentations\xspace}

\newcommand \pol {polynomial\xspace}
\newcommand \pols {polynomials\xspace}

\newcommand \polcar {characteristic \pol}

\newcommand \prc {rank constant \pro}

\newcommand \prmt {précisely\xspace}

\newcommand \prn {projection\xspace}
\newcommand \prns {projections\xspace}

\newcommand \pro {projective\xspace}

\newcommand \proi {potential prime\xspace}
\newcommand \prois {potential primes\xspace}

\newcommand \proc {potential chain\xspace}
\newcommand \procs {potential chains\xspace}

\newcommand \proel {elementary \proc}
\newcommand \proels {elementary \procs}
\newcommand \proelo {\proel of length }
\newcommand \proelos {\proels of length }

\newcommand \prolo {\proc of length }
\newcommand \prolos {\procs of length }

\newcommand \prt {property\xspace}
\newcommand \prts {properties\xspace}

\newcommand \pst {Positivstellensatz\xspace}
\newcommand \psts {Positivstellens\"atze\xspace}

\newcommand \ptf {\tf \pro}


\newcommand \qc {quasi-compact\xspace}

\newcommand \qi {quasi integral\xspace}

\newcommand \rcf {real closed field\xspace}
\newcommand \rcfs {real closed fields\xspace}

\newcommand \rdl {linear dependance relation\xspace}
\newcommand \rdls {linear dependance relations\xspace}

\newcommand \rdi {integral dependance relation\xspace}
\newcommand \rdis {integral dependance relations\xspace}

\newcommand \ri {ring\xspace}
\newcommand \ris {rings\xspace}


\newcommand \sad {dynamical algebraic structure\xspace}
\newcommand \sads {dynamical algebraic structures\xspace}

\newcommand \sdz {without \dvz}
\newcommand \sdzz {without \dvzz}

\newcommand \sgr {\gtr set\xspace}
\newcommand \sgrs {\gtr sets\xspace}

\newcommand \sli {\lin \sys}
\newcommand \slis {\lin \syss}

\newcommand \sys {system\xspace}
\newcommand \syss {systems\xspace}

\newcommand \talg {Horn theory\xspace}
\newcommand \talgs {Horn theories\xspace}

\newcommand \tco {coherent theory\xspace}
\newcommand \tcos {coherent theories\xspace}

\newcommand \tdy {dynamical theory\xspace}
\newcommand \tdys {dynamical theories\xspace}

\newcommand \tel {regular theory\xspace}
\newcommand \tels {regular theories\xspace}

\newcommand \telri {cartesian theory\xspace}
\newcommand \telris {cartesian theories\xspace}

\newcommand \tf {finitely generated\xspace}

\newcommand \tfo {formal theory\xspace}
\newcommand \tfos {theory formelles\xspace}

\newcommand \tgm {\gmq theory\xspace}
\newcommand \tgms {\gmq theories\xspace}

\newcommand \Tho {Theorem\xspace}
\newcommand \tho {theorem\xspace}
\newcommand \thos {theorems\xspace}

\newcommand \tpe {purely equational theory\xspace}

\newcommand \trdi {distributive lattice\xspace}
\newcommand \trdis {distributive lattices\xspace}

\newcommand \vfn {verification\xspace}
\newcommand \vfns {verifications\xspace}

\newcommand \zed {zero-dimensional\xspace}


\newcommand \cov {constructive\xspace}

\newcommand \coma {\cov \maths}
\newcommand \clama {classical \maths}

\renewcommand \cot {constructively\xspace}

\newcommand \mathe {mathematical\xspace}
\newcommand \maths {mathematics\xspace}

\newcommand \matn {mathematician\xspace}

\newcommand \pte {excluded middle principle\xspace}

\newcommand \prco {\cov proof\xspace}
\newcommand \prcos {constructive proofs\xspace}

\newcommand \tcg {compactness theorem\xspace}
\newcommand \Tcgi {The \tcg implies the following result. }
%



\theoremstyle{plain}
\newtheorem{theorem}{Theorem}[section]
\newtheorem{thdef}[theorem]{Theorem and definition}
\newtheorem{lemma}[theorem]{Lemma}
\newtheorem{corollary}[theorem]{Corollary}
\newtheorem{proposition}[theorem]{Proposition}
\newtheorem{propdef}[theorem]{Proposition and definition}
\newtheorem{plcc}[theorem]{Concrete local-global principle}
\newtheorem{fact}[theorem]{Fact}
\newtheorem{algorithm}{Algorithm}

\theoremstyle{definition}
\newtheorem{conjecture}[theorem]{Conjecture}
\newtheorem{definition}[theorem]{Definition}
\newtheorem{definitions}[theorem]{Definitions}
\newtheorem{notation}[theorem]{Notation}
\newtheorem{definota}[theorem]{Definition and notation} 
\newtheorem{convention}[theorem]{Convention}
\newtheorem{problem}[theorem]{Problem}
\newtheorem{question}[theorem]{Question}

\theoremstyle{remark}
\newtheorem{remark}[theorem]{Remark}
\newtheorem{remarks}[theorem]{Remarks}
\newtheorem{comment}[theorem]{Comment}
\newtheorem{comments}[theorem]{Comments}
\newtheorem{example}[theorem]{Example}
\newtheorem{examples}[theorem]{Examples}

\rdb
\label{beginenglish}

\begin{abstract}
We propose a slight modification of the Berlekamp-Massey Algorithm
for obtaining the minimal polynomial of a given linearly recurrent
sequence. Such a modification enables to explain it in a simpler
way and to adapt it to lazy evaluation.
\end{abstract}

\bni MSC 2000: 68W30, 15A03

\mni Key words: Berlekamp-Massey Algorithm. Linearly recurrent
sequences.

\bni {\bf Note:} This paper appeared in AAECC {\bf 17} 1 (2006), 75--82.
Received: 14 September 2005 / Published online: 9 March 2006

\section{Introduction: The usual Berlekamp-Massey algorithm }

Let $\mathbb{K}$ be an arbitrary field. Given a linearly recurrent
sequence, denoted by $\mathcal{S}(x)=\sum_{i=0}^\infty a_i
x^i$, $a_i \in \mathbb{K}$, we wish to compute its minimal
polynomial, denoted by $P(x)$. Recall that if $P(x)$ is given by
$P(x)=\sum_{i=0}^d p_i x^i$ denotes such polynomial, then
$P(x)$ is the polynomial of the smallest degree such that $
\sum_{i=0}^d p_i a_{j+i}=0, $ for all $j$ in $\N$.

Let suppose that the minimal polynomial of $\mathcal{S}(x)$ has
degree bound $n$. Under such hypothesis, the Berlekamp-Massey
Algorithm only requires the first $2n$ coefficients of
$\mathcal{S}(x)$ in order to compute the minimal polynomial. Such
coefficients define the polynomial $S=\sum_{i=0}^{2n-1}a_i\,x^i$.

A large literature can be consulted nowadays in relation to the
Berlekamp's Algorithm.  The (original) Berlekamp's Algorithm was
created for decoding Bose-Chaudhuri-Hocquenghem (BCH) codes in
1968 (see \cite{br}). One year later, the original version of this
algorithm has been simplified by Massey (see \cite{br-ms}). The
similarity of the algorithm to the extended Euclidean Algorithm
can be found in several articles, for instance, in
\cite{cheng}, \cite{dorns}, \cite{milles}, \cite{sugi} and
\cite{welch}. Some more recent interpretations of the
Berlekamp-Massey Algorithm in terms of Hankel Matrices and
Pad\'{e} approximations can be found in \cite{jon} and \cite{pan}.

The usual interpretation of the Berlekamp-Massey Algorithm for
obtaining $P(x)$ is expressed in pseudocode in Algorithm 1.

\begin{algor}[The Usual Berlekamp-Massey Algorithm]\label{BMA}
\Entree  $n \in \N $. The first $2n$ coefficients of a linearly
recurrent sequence defined over $\K$, given by the list
$[a_0,a_1,\ldots,a_{2n-1}]$. The minimal polynomial has degree
bound $n$. 
\Sortie The minimal polynomial $P$ of the sequence.
\Debut 
\Varloc $R,R_0,R_1,V,V_0,V_1,Q$ : \pols in $x$ \hst \#
initialization \hsu $R_0:=x^{2n}$ ;
$R_1:=\sum_{i=0}^{2n-1}a_i\,x^i$ ;  $V_0=0$  ; $V_1=1$ ; \hst \#
loop 
\hsu    \tantque{n \leq \deg(R_1)} \hsd $(Q,R) :=$ quotient
and remainder of $R_0$ divided by $R_1$ ; \hsd         $V :=
V_0-Q\,V_1$ ; \hsd $V_0:=V_1$ ; $V_1:=V$ ; $R_0:=R_1$ ; $R_1:=R$ ;
\hsu \fintantque \hst \# exit \hsu $d:=\max(\deg(V_1),1+\deg(R_1))
$ ; $P:=x^dV_1(1/x)$ ; Return $P:=P/ \mathrm{leadcoeff}(P)$. \fin
\end{algor}

In practice, we must apply the simplification of the extended
Euclidean Algorithm given in \cite{dorns}, to find exactly the
Berlekamp-Massey Algorithm. Such simplification is based on the
fact that initial $R_0$ is equal to $x^{2n}$.

Although Algorithm \ref{BMA} is not complicated, it seems to be no
easy to find a direct and transparent explanation for the
determination of the degree of $P$. In the literature, we think
there is a little confusion with the different definitions of
minimal polynomial and with the different ways of defining the
sequence. Here, we introduce a slight modification of the
algorithm which makes it more comprehensible and natural. We did
not find in the literature such a modification before the first
submission of this article (May 2004). However, we would like to
add that you can also find it in \cite{shoup}, published in 2005,
without any reference.

\section{Some good reasons to modify the usual algorithm}

By the one hand, as it can be observed at the end of Algorithm
\ref{BMA}, we have to compute the (nearly) reverse polynomial of
$V_1$, in order to obtain the right polynomial. The following
example helps us to understand what happens:
$$
\begin{array}{ l}
n\,=\,d\,=\,3,  \\
S\,=\,a_0+a_1x+a_2x^2+a_3x^3+a_4x^4+a_5 x^5\,=\,
1+2x+7x^2-9x^3+2x^4+7x^5,\\ \noalign{\medskip} \mathrm{Algorithm
\, 1}(3,[1,2,7,-9,2,7]) \, \Rightarrow \, P\,=\,x+x^2+x^3,\\
\noalign{\medskip} \mathrm{with\, \,}
V_1\,=\,v_0+v_1x+v_2x^2  \,=\,49/67(1 +  x +   x^2),\\
\noalign{\medskip} \mathrm{and}\,\,R  \,\,\mathrm{such
\,that\,\,}S\,V_1\,=\, R   \mod   x^6,\,\deg(R)\,=\,2\,\\
\noalign{\medskip} \mathrm{which \,\,  implies\,\, that}\\
\noalign{\medskip}
  \begin{array}{ccccccc}
   \mathrm{coeff}(S\,V_1,x,3)&=& a_1v_2+a_2v_1+a_3v_0
   &=& 2v_2+7v_1-9v_0&=&0,\\
\mathrm{coeff}(S\,V_1,x,4)&=&a_2 v_2+a_3v_1+a_4v_0 &=&7 v_2-9v_1+2v_0  &=&0,\\
\mathrm{coeff}(S\,V_1,x,5)&=& a_3v_2 + a_4v_1 + a_5v_0 &=&-9
v_2+2v_1+7v_0 & =& 0.\\
\end{array} \end{array}
$$
Hence, the right degree of $P$ is given by the degree of the last
$R_1$ plus one because $x$ divides $P$. Observe that
$a_0v_2+a_1v_1+a_2v_0 = 490/67 \neq 0$. We would like to obtain
directly the desired polynomial from $V_1$.

Moreover, by the other hand, in Algorithm \ref{BMA} all the first
$2\,n$ coefficients are required to start the usual algorithm,
where $n$ only provides a degree bound for the minimal polynomial.
Consequently, it may be possible that the true degree of $P$ is
much smaller that $n$ and so, less coefficients of the sequence
are required to obtain the wanted polynomial.

So, we suggest a more natural, efficient and direct way to obtain
$P$. Our idea is to consider the polynomial
$\widehat{S}=\sum_{i=0}^{2n-1}a_i\,x^{2n-1-i}$ as the initial
$R_1$. Observe that in this case, using the same notation as in
Algorithm \ref{BMA}, the same example shows that it is not
necessary to reverse the polynomial $V_1$ at the end of the
algorithm.

$$
\begin{array}{ l}
n\,=\,d\,=\,3,  \\
\widehat{S}\,=\,a_0x^5+a_1x^4+a_2x^3+a_3x^2+a_4x +a_5 \,=\,
x^5+2\,x^4+7\,x^3-9\,x^2+2\,x+7,\\
\noalign{\medskip} \mathrm{Algorithm
\, 2}\,(3,[1,2,7,-9,2,7]) \, \Rightarrow \, P\,=\,x+x^2+x^3,\\
\noalign{\medskip} \mathrm{with\, \,}
V_1\,=\,v_0 + v_1x+v_2x^2 +v_3x^3 \,=\,-9/670 ( x + x^2+x^3),\\
\noalign{\medskip} \mathrm{and}\,\,R  \,\,\mathrm{such
\,that\,\,}\widehat{S}\,V_1\,=\, R   \mod   x^6,\,\deg(R)\,=\,2\\
\noalign{\medskip} \mathrm{which \,\,  implies\,\, that}\\
\noalign{\medskip}
\begin{array}{ccccccc}
   \mathrm{coeff}(\widehat{S}\,V_1,x,3)&=& a_2v_0+a_3v_1+a_4v_2+a_5v_3&=&
-9v_1+2v_2+7v_3 &=&0,\\
\mathrm{coeff}(\widehat{S}\,V_1,x,4)&=&a_1v_0+a_2v_1+a_3v_2+a_4v_3&=&7v_1-9v_2+2v_3
&=&0,\\
\mathrm{coeff}(\widehat{S}\,V_1,x,5)&=&
a_0v_0+a_1v_1+a_2v_2+a_3v_3&=&2v_1+7v_2  -9v_3 &=&0.\\
\end{array}
\end{array}
$$
Furthermore, when $ n \gg \mathrm{deg}(P)$, the algorithm can
admit a lazy evaluation. In other words, the algorithm can be
initiated with less coefficients than $2n$ and if the outcome does
not provide the wanted polynomial, we increase the number of
coefficients but remark that it is not necessary to initiate again
the algorithm because we can take advantages of the computations
done before. We will explain this application of the algorithm in
Section \ref{din_eval}.

Next, we introduce our modified Berlekamp-Massey Algorithm in
pseudocode (Algorithm \ref{BMA2}):

\newpage

\begin{algor}[Modified Berlekamp-Massey Algorithm]\label{BMA2}
\Entree  $n \in \N$. The first $2n$ coefficients of a linearly
recurrent sequence defined over $\K$, given by the list
$[a_0,a_1,\ldots,a_{2n-1}]$. The minimal polynomial has degree
bound $n$. \Sortie  The minimal polynomial $P$ of the sequence.
\Debut \Varloc $R,R_0,R_1,V,V_0,V_1,Q$ : \pols in $x$ ; $m=2n-1$ :
integer. \hst \# initialization \hsu $m:=2n-1$ ; $R_0:=x^{2n}$ ;
$R_1:=\sum_{i=0}^m a_{m-i}\,x^i$ ; $V_0=0$ ;  $V_1=1$ ; \hst \#
  loop \hsu    \tantque{n \leq  \deg(R_1)} \hsd        $(Q,R) :=$
quotient and remainder of $R_0$ divided by $R_1$; \hsd    $V
:=V_0-Q\,V_1$ ; \hsd    $V_0:=V_1$ ; $V_1:=V$ ; $R_0:=R_1$ ;
$R_1:=R$ ; \hsu \fintantque \hst \# exit \hsu Return
$P:=V_1/\cd(V_1)$;
  \fin
\end{algor}

Now we prove our result. Let \,$\s{a}=(a_n)_{n\in \N}$\ be an
arbitrary list and \,$i,$ $r,$ $p\in\N$. Let $\rH^{\s{a}}_{i,r,p}$
denote the following Hankel matrix of order $r \times p$,
$$\rH^{\s{a}}_{i,r,p}=
\left[ \begin{array}{ccccc}
a_i & a_{i+1} & a_{i+2}   &\ldots & a_{i+p-1}  \\
a_{i+1} & a_{i+2} &       &       & a_{i+p}  \\
a_{i+2} &     &       &       &     \\
\vdots &  &       &       & \vdots   \\
a_{i+r-1}&a_{i+r}&\ldots&\ldots & a_{i+r+p-2}
\end{array}\right]\,
$$
and let $\rP^{\s{a}}(x)$ be the minimal polynomial of $\s a$.

The next proposition shows the well known relation between the
rank of Hankel matrix and the sequence.

\begin{proposition} \label{propSRL}
Let \,$\s{a}$\, be a \srl. If \,$\s{a}$\, has a generating
polynomial of degree $\leq n$, then the degree $d$ of its minimal
polynomial  $\rP^{\s{a}}$ is equal to the rank of the Hankel
matrix
$$
\rH^{\s{a}}_{0,n,n}= \left[ \begin{array}{ccccccc}
a_0 & a_{1} & a_{2}   &\cdots &  a_{n-2}& a_{n-1} \\
a_{1} & a_{2} &       & \adots  &a_{n-1} & a_n \\
a_{2} &     & \adots & \adots & \vdots  &\vdots    \\
\vdots &\adots   & \adots  &       & \vdots  & \vdots\\
a_{n-2}&a_{n-1}&\cdots&\cdots & a_{2n-2}&a_{2n-1}\\
a_{n-1}&a_n&\cdots&\cdots & a_{2n-1}&a_{2n-2} \\
\end{array}
\right] \; .$$ The coefficients of
\,$\rP^{\s{a}}(x)=x^d-\sum_{i=0}^{d-1}g_ix^i\in\K[x]\,$ are
provided by the unique solution of the linear system
$$
\rH^{\s{a}}_{0,d,d}\,G=\rH^{\s{a}}_{d,d,1},
$$
that is,

\begin{equation} \label{EqSRL3}
\left[ \begin{array}{ccccc}
a_0 & a_1 & a_2   &\cdots & a_{d-1}  \\
a_1 & a_2 &       & \adots  & a_d  \\
a_2 &     & \adots & \adots & \vdots    \\
\vdots &\adots   & \adots  &       & \vdots   \\
a_{d-1}&a_d&\cdots&\cdots & a_{2d-2}
\end{array}\right] \;
\left[ \begin{array}{c}
g_0 \\ g_1 \\ g_2 \\  \vdots   \\  g_{d-1}
\end{array}\right]
=
\left[ \begin{array}{c}
a_d \\  a_{d+1} \\ a_{d+2} \\  \vdots \\ a_{2d-1}
\end{array}\right] \,.
\end{equation}
\end{proposition}

As an immediate corollary of Proposition \ref{propSRL}, we have
the following result.

\begin{corollary}
Using the notation of Proposition \ref{propSRL}, a vector
$Y=(p_0,\ldots,p_n)$ is solution of
$$
\rH^{\s{a}}_{0,n,n+1}\,Y=0,
$$
that is,
\begin{equation} \label{EqSRL5}
\left[ \begin{array}{cccccc}
a_0 & a_1 & a_2   &\cdots & a_{n-1}&a_n  \\
a_1 & a_2 &       & \adots  & a_n& a_{n+1}  \\
a_2 &     & \adots & \adots & \vdots& \vdots    \\
\vdots &\adots   & \adots  &       & \vdots& \vdots   \\
a_{n-1}&a_n&\cdots&\cdots & a_{2n-2} & a_{2n-1}
\end{array}\right] \;
\left[ \begin{array}{c}
p_0 \\ p_1 \\ p_2 \\  \vdots   \\  p_{n-1}\\  p_{n}
\end{array}\right]
= 0
\end{equation}
if and only if the polynomial $P(x)=\sum_{i=0}^n p_ix^i\in\K[x]$
is multiple of $\,\rP^{\s{a}}(x)$.
\end{corollary}

\begin{proof}
By Proposition \ref{propSRL} the dimension of $\Ker(\rH^{\s{a}}_{0,n,n+1})$
is $n-d$.
For $ 0 \leq j \leq n-1 $, let $C_j$ denote the $j$th column of
$\rH^{\s{a}}_{0,n,n+1}$, that is
$C_j=\rH^{\s{a}}_{j,n,1}=[a_j,a_{j+1},\ldots,a_{n+j-1}]^t$. Since
$\,\rP^{\s{a}}(x)$ is a generating polynomial of
$\s{a}$, for $d\leq j\leq n-1$, we obtain that
$$
C_j - \sum\nolimits_{i=j-d}^{j-1}g_{i-j+d}\,C_i=0.
$$
Thus the  linear independent columns
$[-g_0,\ldots,-g_{d-1},1,0,\ldots,0]^t,\ldots,[0,\ldots,0,-g_0,\ldots,-g_{d-1},1]^t$
  define a basis of  $\Ker(\rH^{\s{a}}_{0,n,n+1})$. Therefore,
$Y=(p_0,\ldots,p_n)$ verifies $\rH^{\s{a}}_{0,n,n+1}\,Y=0$ if and
only if the polynomial $P(x)=\sum_{i=0}^n p_ix^i $ is a multiple
of $\,\rP^{\s{a}}(x)$.
\end{proof}

If we consider $m=2n-1$ and $\widehat{S}=\sum_{i=0}^m a_{m-i}x^i$,
by applying Equation (\ref{EqSRL5}) we obtain:

\begin{equation} \label{EqSRL6}
\exists \,R,\,U\in\K[x]\;\mathrm{\, such \; that}\;
\deg(R)<n,\,\deg(P)\leq n \,\,\, \mathrm{and}\,\,\,
P(x)S(x)+U(x)x^{2n}= R(x).
\end{equation}
Hence, it turns out that finding the minimal polynomial of $\s{a}$
is equivalent to solving (\ref{EqSRL6}) for the minimum degree of
$P$. Moreover, it's well known that

\begin{itemize}
\item the extended Euclidean Algorithm, with $x^{2n}$ and
$\widehat{S}$, provides an equality as (\ref{EqSRL6}) when the
first remainder of degree smaller than $<n$ is reached. Let denote
such remainder by $R_k$,

\item if we consider
  other polynomials $P'(x), \, U'(x)$ and $R'(x)$ such that
$P'(x)\widehat{S}(x)+U'(x)x^{2n}= R'(x)$ and $\deg(R') <
\deg(R_{k-1})$, then   $\deg(P')\geq \deg(P)$ and $\deg(U')\geq
\deg(U)$.
\end{itemize}

That proves that our modification of Berlekamp-Massey Algorithm is
right.

\section{Lazy Evaluation}\label{din_eval}

Our modified Berlekamp-Massey Algorithm admits a lazy evaluation,
which may be very useful in solving the following problem.

Let $f(x)\in \K[x]$ be a squarefree polynomial of degree $n$. Let
$B$ be the universal decomposition algebra of $f(x)$, let $A$ be a
quotient algebra of $B$ and $a \in A$. Thus, $A$ is a
zero--dimensional algebra given by
$$A\simeq\K[X_1,\ldots ,X_n]/\gen{f_1,\ldots ,f_n},
$$
where $f_1,\ldots ,f_n$ define a Gr\"obner basis. Our aim is to
compute the minimal polynomial of $a$, or at least, one of its
factors. However, the dimension of $A$, denoted by $m$, over $\K$
as vector space is normally too big to manipulate matrices of
order $m$. Therefore, we apply the idea of Wiedemann's Algorithm,
by computing the coefficients of a linearly recurrent sequence,
$a_t=\phi(x^t)$, where $\phi$ is a linear form over $A$. Moreover,
since the computation of $x^t$ is usually very expensive and the
minimal polynomial is likely to have degree smaller than the
dimension, we are interested in computing the smallest possible
number of coefficients in order to get the wanted polynomial.

Hence, we first choose $l < m$. We start Algorithm \ref{BMA2} with
$l$ and $[\phi(x^0),\ldots,\phi(x^{2l-1})]$ as input, obtaining a
polynomial as a result. Now, we test if such a polynomial is the
minimal one. If this is not the case, we choose again another
$l'$, $l < l'\leq m$, and we repeat the process with $2l'$
coefficients. However, in this next step, it is possible to take
advantages of all the quotients computed before (with the
exception of the last one), such that Euclidean Algorithm starts
  at $R_0=U_0 x^{2l'}+V_0 \sum_{i=0}^{2l'-1}
(\phi(x^{2l'-1-i}) x^i)$ and $R_1=U_1 x^{2l'}+V_1
\sum_{i=0}^{2l'-1}( \phi(x^{2l'-1-i}) x^i)$, where $U_0$,
$V_0$, $U_1$ and $V_1$ are Bezout coefficients computed in the
previous step. Manifestly, repeating this argument again and
again, we obtain the minimal polynomial.

The following pseudocode tries to facilitate the understanding of
our lazy version of Berlekamp-Massey Algorithm.

Obviously, the choice of $l$ is not unique. Here we have started
at $l=m/4$, adding two coefficients in every further step. In
practice, the particular characteristics of the given problem
could help to choose a proper $l$ and the method of increasing it
through the algorithm. Of course, the simplification of the
Euclidean Algorithm in \cite{dorns} must be considered to optimize
the procedure.

\newpage

 \begin{algorH}[The lazy Berlekamp-Massey Algorithm (in some
particular context)]\label{BMAP}
\Entree  $ m \in \N $, $C \in \K^n$, $G$: Gr\"obner basis, $a \in
A$. The minimal polynomial has degree bound $m$.
\Sortie The minimal polynomial $P$ of $a$
\Debut \Varloc $l,i$: integers,
$R,R_{-1},R_0,R_1,V,V_{-1},V_0,V_1,U,U_{-1},U_0,U_1, S_0,S_1,Q$ :
\pols in $x$, $L,W$:lists, \(val\): element of \(A\);
\hst \# initialization
  \hsu
$ l =\lfloor m/4 \rfloor$;
\hsu $L:=[1,a]$;
$W:=[1,\mathrm{Value}(a,C)]$;
\hsu $S_0:=x^{2l}$ ; $S_1 =W[1]\,x^{2l-1}+W[2]\,x^{2l-2}$;
  \hst \# loop
\hsu \pour{i}{3}{2l}\hsd
         $L[i]:=\mathrm{normalf}(L[i-1]a ,G);$
\hsd           $V[i]:=\mathrm{Value}( L[i],C);$
\hsd         $S_1 = S_1 + V[i] x^{2l-i};$
\hsu
\finpour \hsu
$R_0:=S_0; R_1:=S_1$;
  $V_0=0$ ; $V_1=1$ ; $U_0=1$ ; $V_1=0$;
   \hst \# loop \hsu
\tantque{l \leq \deg(R_1)} \hsd $(Q,R) :=$ quotient and remainder
of $R_0$ divided by $R_1$ ; \hsd $V := V_0-Q V_1$; $U := U_0-Q
U_1$ ;  $U_{-1}:= U_0$; $V_{-1}:= V_0$; \hsd $V_0:=V_1$ ; $V_1:=V$
;$U_0:=U_1$ ; $U_1:=U$;
  $R_0:=R_1$ ; $R_1:=R$ ;
\hsu
\fintantque
  \hsu
\(val:=\mathrm{Subs}(x=a,V_1)\);
   \hst \# loop \hsu
\tantque{val\neq 0 } \hsd $l:=l+1;$
  \hsd \# loop
\hsd \pour{i}{2l-1}{2l} \hst
         $L[i]:=\mathrm{normalf}(L[i-1]a ,G)$; \hst
         $W[i]:=\mathrm{Value}( L[i],C)$;
\hsd \finpour 
\hsd
  $S_0=x^2 S_0$; $S_1 = x^2S_1 + W[2l-1] x+ W[2l];$
\hsd
$R_0:=U_{-1} S_0+V_{-1} S_1$; $R_1:=U_0 S_0+V_0 S_1$ ; 
  \hsd
$U_1:=U_0; V_1:=V_0; U_0:=U_{-1}; V_0:=V_{-1};$
  \hsd \# loop \hsd
\tantque{l \leq \deg(R_1)} 
\hst $(Q,R) :=$ quotient and remainder
of $R_0$ divided by $R_1$ ; 
\hst $V := V_0-Q V_1$; $U := U_0-Q
U_1$; $U_{-1}:= U_0$; $V_{-1}:= V_0$; 
\hst $V_0:=V_1$ ; $V_1:=V$
;$U_0:=U_1$ ; $U_1:=U$;
  $R_0:=R_1$ ; $R_1:=R$ ;
\hsd
\fintantque
  \hsd
  \(val:=\mathrm{Subs}(x=a,V_1)\)
\hsu \fintantque
\hst \# exit \hsu  Return $P:=V_1/\mathrm{leadcoeff}(P)$. \fin
\end{algorH}

\endgroup
\stopcontents[english]

\clearpage
\newpage
\thispagestyle{empty}

\clearpage
\newpage

\setcounter{page}{1}\renewcommand\thepage{F\arabic{page}}\renewcommand\theHsection{F\arabic{section}}

\clearpage\setcounter{section}{0}\selectlanguage{french}\def\frenchproofname{\textsl{Démonstration}}


\newcommand\por[3]{\textbf{pour } $#1$ \textbf{ de } $#2$
      \textbf{ \`a } $#3$  }
\newcommand\sialors[1]{\textbf{si } $#1$ \textbf{ alors }}
\newcommand\tantque[1]{\textbf{tant que } $#1$ \textbf{ faire }}
\newcommand\finpour{\textbf{fin pour}}
\newcommand\sinon{\textbf{sinon }}
\newcommand\finsi{\textbf{fin si }}
\newcommand\fintantque{\textbf{fin tant que }}
\newcommand\Debut{\\[1mm] \textbf{D\'ebut }}
\newcommand\fin{\textbf{\\ Fin.}}
\newcommand\Entree{\\[1mm] \textbf{Entr\'ee : }}
\newcommand\Sortie{\\ \textbf{Sortie : }}
\newcommand\Varloc{\\ \textbf{Variables locales : }}



\newcounter{MF}
\newcommand\stMF{\stepcounter{MF}}

\newcommand{\lec}{\stMF\ifodd\value{MF}lecteur \else 
lectrice \fi}
\newcommand{\lecz}{\stMF\ifodd\value{MF}lecteur\else lectrice\fi}

\newcommand{\lecs}{\stMF\ifodd\value{MF}lecteurs \else 
lectrices \fi}
\newcommand{\lecsz}{\stMF\ifodd\value{MF}lecteurs\else 
lectrices\fi}

\newcommand{\alec}{\stMF\ifodd\value{MF}au lecteur \else%
à la lectrice \fi}
\newcommand{\alecz}{\stMF\ifodd\value{MF}au lecteur\else%
à la lectrice\fi}

\newcommand{\dlec}{\stMF\ifodd\value{MF}du lecteur \else%
de la lectrice \fi}
\newcommand{\dlecz}{\stMF\ifodd\value{MF}du lecteur\else%
de la lectrice\fi}

\newcommand{\llec}{\stMF\ifodd\value{MF}le lecteur \else la lectrice \fi}
\newcommand{\llecz}{\stMF\ifodd\value{MF}le lecteur\else la lectrice\fi}

\newcommand{\Llec}{\stMF\ifodd\value{MF}Le lecteur \else La lectrice \fi}

\newcommand{\lui}{\ifodd\value{MF}lui \else
elle \fi}
\newcommand{\luiz}{\ifodd\value{MF}lui\else
elle\fi}

\newcommand{\celui}{\ifodd\value{MF}celui \else
celle \fi}

\newcommand{\ceux}{\ifodd\value{MF}ceux \else
celles \fi}

\newcommand{\er}{\ifodd\value{MF}er \else
ère \fi}

\newcommand{\eux}{\ifodd\value{MF}eux \else
elles \fi}

\newcommand{\eUx}{\ifodd\value{MF}eux \else
euse \fi}

\newcommand{\leux}{\ifodd\value{MF}leux \else
leuse \fi}

\newcommand{\il}{\ifodd\value{MF}il \else
elle \fi}

\newcommand{\ien}{\ifodd\value{MF}ien \else
ienne \fi}

\newcommand{\e}{\ifodd\value{MF} \else e \fi}
\newcommand{\ez}{\ifodd\value{MF}\else e\fi}

\newcommand{\n}{\ifodd\value{MF}n \else nne \fi}
\newcommand{\nz}{\ifodd\value{MF}n\else nne\fi}

\makeatletter
\newcommand{\la}{\@ifstar{\ifodd\value{MF}le\else
la\fi}{\stMF\ifodd\value{MF}le \else la \fi}}
\makeatother

\newcommand \rem{\rdb
\noi{\sl Remarque. }}

\newcommand \REM[1]{\rdb
\noi{\sl Remarque#1. }}

\newcommand \rems{\rdb
\noi{\sl Remarques. }}

\newcommand \exl{\rdb
\noi{\bf Exemple. }}

\newcommand \EXL[1]{\rdb
\noi{\bf Exemple: #1. }}

\newcommand \exls{\rdb
\noi{\bf Exemples. }}

\newcommand \thref[1] {théorème~\ref{#1}}
\newcommand \paref[1] {page~\pageref{#1}}
\newcommand \pstfref[1] {Posi\-tiv\-stel\-lensatz formel~\ref{#1}}
\newcommand \pstref[1] {Posi\-tiv\-stel\-lensatz~\ref{#1}}

\newcommand\oge{\leavevmode\raise.3ex\hbox{$\scriptscriptstyle\langle\!\langle\,$}}
\newcommand\feg{\leavevmode\raise.3ex\hbox{$\scriptscriptstyle\,\rangle\!\rangle$}}

\newcommand\gui[1]{\oge{#1}\feg}

\newcommand \facile{\begin{proof}
La démonstration est laissée \alecz.
\end{proof}
}

\newcommand\comm{\rdb
\noi{\it Commentaire. }}

\newcommand\COM[1]{\rdb
\noi{\it Commentaire #1. }}

\newcommand\comms{\rdb
\noi{\it Commentaires. }}

\newcommand\Pb{\rdb
\noi{\bf Problème. }}

\newcommand\eoq{\hbox{}\nobreak
\vrule width 1.4mm height 1.4mm depth 0mm}

\newcommand \Cad {C'est-à-dire\xspace}
\newcommand \recu {récur\-rence\xspace}
\newcommand \hdr {hypo\-thèse de \recu}
\newcommand \cad {c'est-à-dire\xspace}
\newcommand \cade {c'est-à-dire en\-co\-re\xspace}
\newcommand \ssi {si, et seu\-lement si, }
\newcommand \ssiz {si, et seu\-lement si,~}
\newcommand \cnes {con\-di\-tion néces\-sai\-re et suf\-fi\-san\-te\xspace}
\newcommand \spdg {sans per\-te de géné\-ra\-lité\xspace}
\newcommand \Spdg {Sans per\-te de géné\-ra\-lité\xspace}

\newcommand \Propeq {Les pro\-pri\-é\-tés sui\-van\-tes sont 
équi\-va\-len\-tes.}
\newcommand \propeq {les pro\-pri\-é\-tés sui\-van\-tes sont 
équi\-va\-len\-tes.}

\newcommand \Kev {$\gK$-\evc}
\newcommand \Kevs {$\gK$-\evcs}

\newcommand \Lev {$\gL$-\evc}
\newcommand \Levs {$\gL$-\evcs}

\newcommand \Qev {$\QQ$-\evc}
\newcommand \Qevs {$\QQ$-\evcs}

\newcommand \kev {$\gk$-\evc}
\newcommand \kevs {$\gk$-\evcs}

\newcommand \lev {$\gl$-\evc}
\newcommand \levs {$\gl$-\evcs}

\newcommand \Alg {$\gA$-\alg}
\newcommand \Algs {$\gA$-\algs}

\newcommand \Blg {$\gB$-\alg}
\newcommand \Blgs {$\gB$-\algs}

\newcommand \Clg {$\gC$-\alg}
\newcommand \Clgs {$\gC$-\algs}

\newcommand \klg {$\gk$-\alg}
\newcommand \klgs {$\gk$-\algs}

\newcommand \llg {$\gl$-\alg}
\newcommand \llgs {$\gl$-\algs}

\newcommand \Klg {$\gK$-\alg}
\newcommand \Klgs {$\gK$-\algs}

\newcommand \Llg {$\gL$-\alg}
\newcommand \Llgs {$\gL$-\algs}

\newcommand \QQlg {$\QQ$-\alg}
\newcommand \QQlgs {$\QQ$-\algs}

\newcommand \Rlg {$\gR$-\alg}
\newcommand \Rlgs {$\gR$-\algs}

\newcommand \RRlg {$\RR$-\alg}
\newcommand \RRlgs {$\RR$-\algs}

\newcommand \ZZlg {$\ZZ$-\alg}
\newcommand \ZZlgs {$\ZZ$-\algs}

\newcommand \Amo {$\gA$-mo\-du\-le\xspace}
\newcommand \Amos {$\gA$-mo\-du\-les\xspace}

\newcommand \Bmo {$\gB$-mo\-du\-le\xspace}
\newcommand \Bmos {$\gB$-mo\-du\-les\xspace}

\newcommand \Cmo {$\gC$-mo\-du\-le\xspace}
\newcommand \Cmos {$\gC$-mo\-du\-les\xspace}

\newcommand \kmo {$\gk$-mo\-du\-le\xspace}
\newcommand \kmos {$\gk$-mo\-du\-les\xspace}

\newcommand \Kmo {$\gK$-mo\-du\-le\xspace}
\newcommand \Kmos {$\gK$-mo\-du\-les\xspace}

\newcommand \Lmo {$\gL$-mo\-du\-le\xspace}
\newcommand \Lmos {$\gL$-mo\-du\-les\xspace}

\newcommand \Ali {appli\-ca\-tion $\gA$-\lin}
\newcommand \Alis {appli\-ca\-tions $\gA$-\lins}

\newcommand \Kli {appli\-ca\-tion $\gK$-\lin}
\newcommand \Klis {appli\-ca\-tions $\gK$-\lins}

\newcommand \Bli {appli\-ca\-tion $\gB$-\lin}
\newcommand \Blis {appli\-ca\-tions $\gB$-\lins}

\newcommand \Cli {appli\-ca\-tion $\gC$-\lin}
\newcommand \Clis {appli\-ca\-tions $\gC$-\lins}

\newcommand \ac{algé\-bri\-quement clos\xspace}  

\newcommand \acl {an\-neau \icl}
\newcommand \acls {an\-neaux \icl}

\newcommand \adp {an\-neau de Pr\"u\-fer\xspace}
\newcommand \adps {an\-neaux de Pr\"u\-fer\xspace}

\newcommand \adpc {\adp \coh}
\newcommand \adpcs {\adps \cohs}

\newcommand \adu {\alg de décom\-po\-sition univer\-selle\xspace}
\newcommand \adus {\algs de décom\-po\-sition univer\-selle\xspace}

\newcommand \adv {anneau de valuation\xspace}
\newcommand \advs {anneaux de valuation\xspace}

\newcommand \advl {anneau \dvla} 
\newcommand \advls {anneaux \dvlas} 

\newcommand \Afr {Anneau \frl}
\newcommand \Afrs {Anneaux \frls}
\newcommand \afr {anneau \frl}
\newcommand \aFr {\hyperref[theorieAfr]{anneau \frl}\xspace}
\newcommand \afrs {anneaux \frls}

\newcommand \afrr {\afr réduit\xspace}
\newcommand \afrrs {\afrs réduits\xspace}
\newcommand \Afrrs {\Afrs réduits\xspace}

\newcommand \afrvr {\afr avec \ravs}
\newcommand \aFrvr {\hyperref[theorieAfrrv]{\afrvr}\xspace}
\newcommand \afrvrs {\afrs avec \ravs}

\newcommand \aftr {anneau réticulé \ftm réel\xspace}
\newcommand \aftrs {anneaux réticulés \ftm réels\xspace}

\newcommand \aG {\alg galoisienne\xspace}
\newcommand \aGs {\algs galoisiennes\xspace}

\newcommand \agB {\alg de Boole\xspace}
\newcommand \agBs {\algs de Boole\xspace}

\newcommand \agH {\alg de Heyting\xspace}
\newcommand \agHs {\algs de Heyting\xspace}

\newcommand \agq{algé\-bri\-que\xspace}
\newcommand \agqs{algé\-bri\-ques\xspace}

\newcommand \agqt{algé\-bri\-que\-ment\xspace}

\newcommand \aKr {anneau de Krull\xspace}
\newcommand \aKrs {anneaux de Krull\xspace}

\newcommand \alg {algè\-bre\xspace}
\newcommand \algs {algè\-bres\xspace}

\newcommand \algo{algo\-rithme\xspace}
\newcommand \algos{algo\-rithmes\xspace}

\newcommand \algq{al\-go\-rith\-mi\-que\xspace}
\newcommand \algqs{al\-go\-rith\-mi\-ques\xspace}

\newcommand \ali {appli\-ca\-tion \lin}
\newcommand \alis {appli\-ca\-tions \lins}

\newcommand \alo {an\-neau lo\-cal\xspace}
\newcommand \alos {an\-neaux lo\-caux\xspace}

\newcommand \algb {an\-neau \lgb}
\newcommand \algbs {an\-neaux \lgbs}

\newcommand \alrd {\alo \dcd}
\newcommand \alrds {\alos \dcds}

\newcommand \anar {anneau \ari}
\newcommand \anars {anneaux \aris}

\newcommand \anor {an\-neau nor\-mal\xspace}
\newcommand \anors {an\-neaux nor\-maux\xspace}

\newcommand \apf {\alg \pf}
\newcommand \apfs {\algs \pf}

\newcommand \apG {\alg pré\-galoisienne\xspace}
\newcommand \apGs {\algs pré\-galoisiennes\xspace}

\newcommand \arc {anneau réel clos\xspace}
\newcommand \aRc {\hyperref[theorieArc]{\arc}\xspace}
\newcommand \arcs {anneaux réels clos\xspace}

\newcommand \ari{arith\-mé\-tique\xspace}  
\newcommand \aris{arith\-mé\-tiques\xspace}  

\newcommand \Asr {Anneau \str}
\newcommand \Asrs {Anneaux \strs}
\newcommand \asr {anneau \str}
\newcommand \asrs {anneaux \strs}

\newcommand \asrvr {\asr avec \ravs}
\newcommand \asrvrs {\asrs avec \ravs}

\newcommand \atf {\alg \tf}
\newcommand \atfs {\algs \tf}

\newcommand \auto {auto\-mor\-phisme\xspace}
\newcommand \autos {auto\-mor\-phismes\xspace}


\newcommand \bdg {base de Gr\"obner\xspace}
\newcommand \bdgs {bases de Gr\"obner\xspace}

\newcommand \bdp {base de \dcn partielle\xspace}
\newcommand \bdps {bases de \dcn partielle\xspace}

\newcommand \bdf {base de \fap\xspace}

\newcommand \Bif {Borne infé\-rieure\xspace} %
\newcommand \bif {borne infé\-rieure\xspace} %
\newcommand \bifs {bornes infé\-rieures\xspace} %

\newcommand \bsp {borne supé\-rieure\xspace} %
\newcommand \bsps {borne supé\-rieures\xspace} %


\newcommand \cac{corps \ac}  

\newcommand \calf{calcul formel\xspace}  

\newcommand \cara{carac\-té\-ris\-tique\xspace}  
\newcommand \caras{carac\-té\-ris\-tiques\xspace}  

\newcommand \carn{carac\-té\-ri\-sation\xspace}  
\newcommand \carns{carac\-té\-ri\-sations\xspace}  

\newcommand \carar{carac\-té\-riser\xspace}

\newcommand \carf{de carac\-tère fini\xspace}  

\newcommand \cdi{corps discret\xspace}
\newcommand \cdis{corps discrets\xspace}
  
\newcommand \cdv{changement de variables\xspace}  
\newcommand \cdvs{changements de variables\xspace}  

\newcommand \cli {clô\-ture inté\-grale\xspace}

\newcommand \codi {corps ordonné discret\xspace}
\newcommand \codis {corps ordonnés discrets\xspace}

\newcommand \coe {coef\-fi\-cient\xspace}
\newcommand \coes {coef\-fi\-cients\xspace}

\newcommand \coh {co\-hé\-rent\xspace}
\newcommand \cohs {co\-hé\-rents\xspace}

\newcommand \cohc {co\-hé\-rence\xspace}

\newcommand \coli {combi\-nai\-son \lin}
\newcommand \colis {combi\-nai\-sons \lins}

\newcommand \com {co\-ma\-xi\-maux\xspace}
\newcommand \come {co\-ma\-xi\-males\xspace}

\newcommand \coo {coor\-donnée\xspace}
\newcommand \coos {coor\-données\xspace}

\newcommand \cop {complé\-men\-taire\xspace}
\newcommand \cops {complé\-men\-taires\xspace}

\newcommand \cosv {conser\-vative\xspace}
\newcommand \cosvs {conser\-vatives\xspace}

\newcommand \cOsv {\hyperref[defithconserv]{conser\-vative}\xspace}
\newcommand \cOsvs {\hyperref[defithconserv]{conser\-vatives}\xspace}

\newcommand \covr {corps ordonné avec \ravs}
\newcommand \covrs {corps ordonnés avec \ravs}

\newcommand \cpb {compa\-tible\xspace} 
\newcommand \cpbs {compa\-tibles\xspace} 

\newcommand \cpbt {compa\-tibi\-lité\xspace} 
\newcommand \cpbtz {compa\-tibi\-lité} 

\newcommand \crc {corps réel clos\xspace}
\newcommand \crcs {corps réels clos\xspace}

\newcommand \crcd {corps réel clos discret\xspace}
\newcommand \crcds {corps réels clos discrets\xspace}


\newcommand \dcd {rési\-duel\-lement dis\-cret\xspace}
\newcommand \dcds {rési\-duel\-lement dis\-crets\xspace}

\newcommand \dcn {décom\-po\-sition\xspace}
\newcommand \dcns {décom\-po\-sitions\xspace}

\newcommand \dcnb {\dcn bornée\xspace}

\newcommand \dcnc {\dcn complète\xspace}

\newcommand \dcnp {\dcn partielle\xspace}

\newcommand \dcp {décom\-posa\-ble\xspace}
\newcommand \dcps {décom\-posa\-bles\xspace}

\newcommand \ddk {dimension de~Krull\xspace}
\newcommand \ddi {de dimension infé\-rieure ou égale à~}

\newcommand \ddp {domaine de Pr\"u\-fer\xspace}
\newcommand \ddps {domaines de Pr\"u\-fer\xspace}

\newcommand \Demo{Démon\-stra\-tion\xspace}     

\newcommand \demo{démon\-stra\-tion\xspace}     
\newcommand \demos{démon\-stra\-tions\xspace}     

\newcommand \dems{démons\-tra\-tions\xspace}

\newcommand \deno{déno\-mi\-nateur\xspace}     
\newcommand \denos{déno\-mi\-nateurs\xspace}   

\newcommand \deter {déter\-mi\-nant\xspace}  
\newcommand \deters {déter\-mi\-nants\xspace}  
  
\newcommand \Dfn{Défi\-nition\xspace}  
\newcommand \Dfns{Défi\-nitions\xspace}  
\newcommand \dfn{défi\-nition\xspace}  
\newcommand \dfns{défi\-nitions\xspace}  

\newcommand \dftr {droite réticulée \ftm réelle\xspace}
\newcommand \dftrs {droites réticulées \ftm réelles\xspace}
  
\newcommand \dil{diffé\-rentiel\xspace}  
\newcommand \dils{diffé\-rentiels\xspace}  
\newcommand \dile{diffé\-ren\-tielle\xspace}  
\newcommand \diles{diffé\-ren\-tielles\xspace}  

\newcommand \dip{diviseur principal\xspace}
\newcommand \dips{diviseurs principaux\xspace}

\newcommand \discri{discri\-minant\xspace}  
\newcommand \discris{discri\-minants\xspace}  

\newcommand \divle {dimension divisorielle\xspace}

\newcommand \dit{distri\-bu\-ti\-vité\xspace}

\newcommand \dlg{d'élar\-gis\-sement\xspace}  

\newcommand \dok {domaine de Dedekind\xspace}
\newcommand \doks {domaines de Dedekind\xspace}

\newcommand \dvla {à diviseurs\xspace}
\newcommand \dvlas {à diviseurs\xspace}

\newcommand \dvld {\dvlt décom\-posé\xspace} %
\newcommand \dvlds {\dvlt décom\-posés\xspace} %

\newcommand \dvlg {diviso\-riel\xspace} 
\newcommand \dvlgs {diviso\-riels\xspace} 

\newcommand \dvli {\dvlt inver\-sible\xspace} 
\newcommand \dvlis {\dvlt inver\-sibles\xspace} 

\newcommand \dvlt {diviso\-riel\-lement\xspace} %

\newcommand \dvz {di\-viseur de zéro\xspace}
\newcommand \dvzs {di\-viseurs de zéro\xspace}

\newcommand \dve {divi\-si\-bi\-lité\xspace}

\newcommand \dvee {à \dve explicite\xspace}

\newcommand \dvr {diviseur\xspace}
\newcommand \dvrs {diviseurs\xspace}


\newcommand \Eds {Exten\-sion des sca\-laires\xspace}
\newcommand \edss {exten\-sions des sca\-laires\xspace}
\newcommand \eds {exten\-sion des sca\-laires\xspace}

\newcommand \eco {\elts \com}

\newcommand \egmt {éga\-lement\xspace}

\newcommand \egt {éga\-li\-té\xspace}
\newcommand \egts {éga\-li\-tés\xspace}

\newcommand \eli{élimi\-nation\xspace}  

\newcommand \elr{élé\-men\-taire\xspace}  
\newcommand \elrs{élé\-men\-taires\xspace}  

\newcommand \elrt{élé\-men\-tai\-rement\xspace}  

\newcommand \elt{élé\-ment\xspace}  
\newcommand \elts{élé\-ments\xspace}  

\def \endo {en\-do\-mor\-phisme\xspace}
\def \endos {en\-do\-mor\-phismes\xspace}

\newcommand \entrel {rela\-tion impli\-ca\-tive\xspace}
\newcommand \entrels {rela\-tions impli\-ca\-tives\xspace}

\newcommand\evc{es\-pa\-ce vec\-to\-riel\xspace} 
\newcommand\evcs{es\-pa\-ces vec\-to\-riels\xspace} 

\newcommand \eqv {équi\-valent\xspace}  
\newcommand \eqve {équi\-va\-lente\xspace}  
\newcommand \eqvs {équi\-valents\xspace}  
\newcommand \eqves {équi\-val\-entes\xspace}  

\newcommand \eqvc {équi\-va\-lence\xspace}  
\newcommand \eqvcs {équi\-va\-lences\xspace}  

\newcommand \esid {essen\-tiel\-lement iden\-tique\xspace}  
\newcommand \esids {essen\-tiel\-lement iden\-tiques\xspace}  

\newcommand \Esid {\hyperref[defitdyesidentiques]{\esid}\xspace}  
\newcommand \Esids {\hyperref[defitdyesidentiques]{\esids}\xspace}  

\newcommand \eseq {essen\-tiel\-lement \eqve}  
\newcommand \eseqs {essen\-tiel\-lement \eqves}  

\newcommand \Eseq {\hyperref[defitheseq]{\eseq}\xspace}  
\newcommand \Eseqs {\hyperref[defitheseq]{\eseqs}\xspace}

\newcommand \fab {\fcn bornée\xspace}
\newcommand \fabs {\fcns bornées\xspace}

\newcommand \fat {\fcn totale\xspace}
\newcommand \fats {\fcn totales\xspace}

\newcommand \fap {\fcn partielle\xspace}
\newcommand \faps {\fcns partielles\xspace}

\newcommand \fip {filtre pre\-mier\xspace}
\newcommand \fips {filtres pre\-miers\xspace}

\newcommand \fipma {\fip maxi\-mal\xspace}
\newcommand \fipmas {\fips maxi\-maux\xspace}

\newcommand \fcn {factorisation\xspace}
\newcommand \fcns {factorisations\xspace}

\newcommand \fdi {for\-te\-ment dis\-cret\xspace}
\newcommand \fdis {for\-te\-ment dis\-crets\xspace}

\newcommand \fsa {fermé \sagq}
\newcommand \fsas {fermés \sagqs}

\newcommand \fsagc {fonction \sagc}
\newcommand \fsagcs {fonctions \sagcs}

\newcommand \fmt {formel\-lement\xspace}

\newcommand \frl {for\-tement réticulé\xspace}
\newcommand \frle {for\-tement réticulée\xspace}
\newcommand \frls {for\-tement réticulés\xspace}

\newcommand \ftm {fortement\xspace}

\newcommand\gmt{géométrie\xspace}  
\newcommand\gmts{géométries\xspace}  

\newcommand\gaq{\gmt \agq}  

\newcommand\gmq{géomé\-trique\xspace}  
\newcommand\gmqs{géomé\-triques\xspace}  

\newcommand\gmqt{géomé\-tri\-quement\xspace}  

\newcommand\gne{géné\-ra\-lisé\xspace}  
\newcommand\gnee{géné\-ra\-lisée\xspace}  
\newcommand\gnes{géné\-ra\-lisés\xspace}  
\newcommand\gnees{géné\-ra\-lisées\xspace}  

\newcommand\gnl{géné\-ral\xspace}  
\newcommand\gnle{géné\-rale\xspace}  
\newcommand\gnls{géné\-raux\xspace}  
\newcommand\gnles{géné\-rales\xspace}  

\newcommand\gnlt{géné\-ra\-lement\xspace}  

\newcommand\gnn{géné\-ra\-li\-sa\-tion\xspace}  
\newcommand\gnns{géné\-ra\-li\-sa\-tions\xspace}  

\newcommand\gnq {géné\-rique\xspace}  
\newcommand\gnqs {géné\-riques\xspace}  

\newcommand\gnr{géné\-ra\-liser\xspace}  

\newcommand \gns{géné\-ra\-lise\xspace}

\newcommand \gnt{géné\-ra\-lité\xspace}
\newcommand \gnts{géné\-ra\-lités\xspace}

\newcommand \grl{groupe \rtl}
\newcommand \grls{groupes \rtls}

\newcommand \gRl {\hyperref[theorieGrl]{\grl}\xspace}
\newcommand \gRls {\hyperref[theorieGrl]{\grls}\xspace}

\newcommand\gtr{géné\-ra\-teur\xspace}  
\newcommand\gtrs{géné\-ra\-teurs\xspace}  


\newcommand \homo {ho\-mo\-mor\-phisme\xspace}
\newcommand \homos {ho\-mo\-mor\-phismes\xspace}

\newcommand \hmg {homo\-gène\xspace}
\newcommand \hmgs {homo\-gènes\xspace}

\newcommand \icftr {intervalle compact réticulé \ftm réel\xspace}
\newcommand \icftrs {intervalles compacts réticulés \ftm réels\xspace}

\newcommand \icl {inté\-gra\-lement clos\xspace}
\newcommand \icle {inté\-gra\-lement close\xspace}

\newcommand \icsr {intervalle compact \stm réticulé\xspace}
\newcommand \icsrs {intervalles compacts \stm réticulés\xspace}

\newcommand \icrc {intervalle compact réel clos\xspace}
\newcommand \icrcs {intervalles compact réels clos\xspace}

\newcommand \id {idéal\xspace}
\newcommand \ids {idéaux\xspace}

\newcommand \ida {\idt \agq}
\newcommand \idas {\idts \agqs}

\newcommand \idc  {\idt de Cramer\xspace}
\newcommand \idcs {\idts de Cramer\xspace}

\newcommand \idd {idéal déter\-minan\-tiel\xspace}
\newcommand \idds {idéaux déter\-minan\-tiels\xspace}

\newcommand \idema {idéal maxi\-mal\xspace}
\newcommand \idemas {idéaux maxi\-maux\xspace}

\newcommand \idep {idéal pre\-mier\xspace}
\newcommand \ideps {idéaux pre\-miers\xspace}

\newcommand \idemi {\idep minimal\xspace}
\newcommand \idemis {\ideps minimaux\xspace}

\newcommand \idf {idéal de Fitting\xspace}
\newcommand \idfs {idéaux de Fitting\xspace}

\newcommand \idif {idéal \dvlg fini\xspace}
\newcommand \idifs {idéaux \dvlgs finis\xspace}

\newcommand \idli {idéal \dvli\xspace} 
\newcommand \idlis {idéaux \dvlis\xspace} 

\newcommand \idm {idem\-potent\xspace}
\newcommand \idms {idem\-potents\xspace}
\newcommand \idme {idem\-potente\xspace}
\newcommand \idmes {idem\-potentes\xspace}

\newcommand \idp {idéal prin\-ci\-pal\xspace}
\newcommand \idps {idé\-aux prin\-ci\-paux\xspace}

\newcommand \idt {iden\-ti\-té\xspace}
\newcommand \idts {iden\-ti\-tés\xspace}

\newcommand \idtr {indé\-ter\-minée\xspace}
\newcommand \idtrs {indé\-ter\-minées\xspace}

\newcommand \ifr {idéal frac\-tion\-naire\xspace}
\newcommand \ifrs {idéaux frac\-tion\-naires\xspace}

\newcommand \imd {immé\-diat\xspace}
\newcommand \imde {immé\-diate\xspace}
\newcommand \imds {immé\-diats\xspace}
\newcommand \imdes {immé\-diates\xspace}

\newcommand \imdt {immé\-dia\-te\-ment\xspace}

\newcommand \indtr {inf-demi-treillis\xspace} 

\newcommand \inteq {intui\-ti\-vement \eqve}
\newcommand \inteqs {intui\-ti\-vement \eqves}

\newcommand \Inteq {\hyperref[defextintequiv]{\inteq}\xspace}
\newcommand \Inteqs {\hyperref[defextintequiv]{\inteqs}\xspace}

\newcommand \ing {in\-ver\-se \gne}
\newcommand \ings {in\-ver\-ses \gnes}

\newcommand \iMP {in\-ver\-se de Moo\-re-Pen\-ro\-se\xspace}
\newcommand \iMPs {in\-ver\-ses de Moo\-re-Pen\-ro\-se\xspace}

\newcommand \ipp {\idep poten\-tiel\xspace}
\newcommand \ipps {\ideps poten\-tiels\xspace}

\newcommand \ird {irré\-duc\-tible\xspace}
\newcommand \irds {irré\-duc\-tibles\xspace}

\newcommand \iso {iso\-mor\-phisme\xspace}
\newcommand \isos {iso\-mor\-phismes\xspace}

\newcommand \itf {idéal \tf}
\newcommand \itfs {idé\-aux \tf}

\newcommand \itid {intui\-ti\-vement iden\-tique\xspace}
\newcommand \itids {intui\-ti\-vement iden\-tiques\xspace}

\newcommand \iv {inversible\xspace}
\newcommand \ivs {inversibles\xspace}

\newcommand \ivdg {inverse divisoriel\xspace} 
\newcommand \ivdgs {inverses divisoriels\xspace} 

\newcommand \ivde {inverse divisorielle\xspace} 
\newcommand \ivdes {inverses divisorielles\xspace} 

\newcommand \ivda {inverse divisoriel\xspace} 
\newcommand \ivdas {inverses divisoriels\xspace} 


\newcommand \lgb {local-global\xspace}
\newcommand \lgbe {locale-globale\xspace}
\newcommand \lgbs {local-globals\xspace}

\newcommand \lin {liné\-aire\xspace}
\newcommand \lins {liné\-aires\xspace}

\newcommand \lint {liné\-ai\-rement\xspace}

\newcommand \lmo {\lot mono\-gène\xspace}
\newcommand \lmos {\lot mono\-gènes\xspace}

\newcommand \lnl {\lot \nl}
\newcommand \lnls {\lot \nls}

\newcommand \lot {loca\-lement\xspace}

\newcommand \lon {loca\-li\-sation\xspace}
\newcommand \lons {loca\-li\-sations\xspace}

\newcommand \lop {\lot prin\-cipal\xspace}
\newcommand \lops {\lot prin\-cipaux\xspace}

\newcommand \lsdz {\lot \sdz}

\newcommand \mdi {mo\-dule des \diles}

\newcommand \mlm {mo\-dule \lmo}
\newcommand \mlms {mo\-dules \lmos}

\newcommand \mlmo {ma\-tri\-ce de loca\-li\-sation
mono\-gène\xspace}
\newcommand \mlmos {ma\-tri\-ces de loca\-li\-sation
mono\-gène\xspace}

\newcommand \mlp {ma\-tri\-ce de loca\-li\-sation
prin\-ci\-pa\-le\xspace}
\newcommand \mlps {ma\-tri\-ces de loca\-li\-sation
prin\-ci\-pa\-le\xspace}

\newcommand \mo {mo\-no\"{\i}de\xspace}
\newcommand \mos {mo\-no\"{\i}des\xspace}

\newcommand \moco {\mos \com}

\newcommand \molo {morphisme de \lon\xspace}
\newcommand \molos {morphismes de \lon\xspace}

\newcommand \mom {mo\-nô\-me\xspace}
\newcommand \moms {mo\-nô\-mes\xspace}

\newcommand \moquo {morphisme de passage au quotient\xspace}
\newcommand \moquos {morphismes de passage au quotient\xspace}

\newcommand \mpf {mo\-dule \pf}
\newcommand \mpfs {mo\-dules \pf}

\newcommand \mpl {mo\-dule plat\xspace}
\newcommand \mpls {mo\-dules plats\xspace}

\newcommand \mpn {ma\-trice de \pn}
\newcommand \mpns {ma\-trices de \pn}

\newcommand \mprn {ma\-trice de \prn}
\newcommand \mprns {ma\-trices de \prn}

\newcommand \mptf {mo\-dule \ptf}
\newcommand \mptfs {mo\-dules \ptfs}

\newcommand \mrc {mo\-dule \prc}
\newcommand \mrcs {mo\-dules \prcs}

\newcommand \mtf {mo\-du\-le \tf}
\newcommand \mtfs {mo\-du\-les \tf}


\newcommand \ncr{néces\-saire\xspace}       
\newcommand \ncrs{néces\-saires\xspace}       

\newcommand \ncrt{néces\-sai\-rement\xspace}       

\newcommand \ndz {régu\-lier\xspace}
\newcommand \ndzs {régu\-liers\xspace}

\newcommand \nl {simple\xspace}
\newcommand \nls {simples\xspace}

\newcommand \noco {\noe \coh}
\newcommand \nocos {\noes \cohs}

\newcommand \Noe {Noether\xspace}

\newcommand \noe {noethé\-rien\xspace}
\newcommand \noes {noethé\-riens\xspace}
\newcommand \noee {noethé\-rienne\xspace}
\newcommand \noees {noethé\-riennes\xspace}

\newcommand \noet {noethé\-ria\-nité\xspace}

\newcommand \nst {Null\-stellen\-satz\xspace}
\newcommand \nsts {Null\-stellen\-s\"atze\xspace}

\newcommand \op{opé\-ra\-tion\xspace}  
\newcommand \ops{opé\-ra\-tions\xspace}
\newcommand \opari{\op\ari}  
\newcommand \oparis{\ops\aris}  
\newcommand \oparisv{\ops\arisv}  

\newcommand \oqc {ouvert \qc}
\newcommand \oqcs {ouverts \qcs}

\newcommand \ort{or\-tho\-go\-nal\xspace}  
\newcommand \orte{or\-tho\-go\-na\-le\xspace}  
\newcommand \orts{or\-tho\-go\-naux\xspace}  
\newcommand \ortes{or\-tho\-go\-na\-les\xspace}  


\newcommand \pa {couple saturé\xspace}
\newcommand \pas {couples saturés\xspace}
 
\newcommand \paral{paral\-lèle\xspace}  
\newcommand \parals{paal\-lèles\xspace}  

\newcommand \paralm{paral\-lè\-lement\xspace}

\newcommand \pb{pro\-blè\-me\xspace}  
\newcommand \pbs{pro\-blè\-mes\xspace}  

\newcommand \peq {purement équa\-tion\-nelle\xspace}
\newcommand \peqs {purement équa\-tion\-nelles\xspace}

\newcommand \pf {de \pn finie\xspace}

\newcommand \plc {rési\-duel\-lement \zed}
\newcommand \plcs {rési\-duel\-lement \zeds}

\newcommand \Plg {Prin\-cipe \lgb}
\newcommand \plg {prin\-cipe \lgb}
\newcommand \plgs {prin\-cipes \lgbs}

\newcommand \plga {\plg abs\-trait\xspace}
\newcommand \plgas {\plgs abs\-traits\xspace}

\newcommand \Plgc {\Plg con\-cret\xspace}
\newcommand \plgc {\plg con\-cret\xspace}
\newcommand \plgcs {\plgs con\-crets\xspace}

\newcommand \pn {présen\-ta\-tion\xspace}
\newcommand \pns {présen\-ta\-tions\xspace}

\newcommand \pog {\pol \hmg\xspace}
\newcommand \pogs {\pols \hmgs\xspace}

\newcommand \Pol {Poly\-nôme\xspace}
\newcommand \Pols {Poly\-nômes\xspace}

\newcommand \pol {poly\-nôme\xspace}
\newcommand \pols {poly\-nômes\xspace}

\newcommand \poll{poly\-nomial\xspace}  
\newcommand \polls{poly\-nomiaux\xspace}  
\newcommand \polle{poly\-no\-miale\xspace}  
\newcommand \polles{poly\-no\-miales\xspace}  

\newcommand \pollt{poly\-no\-mia\-lement\xspace}  

\newcommand \polfon {\pol fon\-da\-men\-tal\xspace}
\newcommand \polmu {\pol rang\xspace}
\newcommand \polmus {\pols rang\xspace}
\newcommand \polcar {\pol carac\-té\-ris\-tique\xspace}
\newcommand \polmin {\pol mini\-mal\xspace}

\newcommand \polg {\pol g\'en\'e\-rateur\xspace}
\newcommand \polgs {\pols g\'en\'e\-rateurs\xspace}
\newcommand \polgmin {\polg minimal\xspace}

\newcommand \prc {\pro de rang constant\xspace}
\newcommand \prcs {\pros de rang constant\xspace}

\newcommand \prcc {prin\-ci\-pe de \rcc}
\newcommand \prca {prin\-ci\-pe de \rca}
\newcommand \prce {prin\-ci\-pe de \rce}

\newcommand \prmt {préci\-sé\-ment\xspace}
\newcommand \Prmt {Préci\-sé\-ment\xspace}

\newcommand \prn {pro\-jec\-tion\xspace}
\newcommand \prns {pro\-jec\-tions\xspace}

\newcommand \pro {pro\-jec\-tif\xspace}
\newcommand \pros {pro\-jec\-tifs\xspace}

\newcommand \prr {pro\-jec\-teur\xspace}
\newcommand \prrs {pro\-jec\-teurs\xspace}

\newcommand \Prt {Pro\-pri\-été\xspace}
\newcommand \Prts {Pro\-pri\-étés\xspace}
\newcommand \prt {pro\-pri\-été\xspace}
\newcommand \prts {pro\-pri\-étés\xspace}

\newcommand \ptf {\pro \tf}
\newcommand \ptfs {\pros \tf}

\newcommand \qc {quasi-compact\xspace}
\newcommand \qcs {quasi-compacts\xspace}

\newcommand \qi {qua\-si in\-tè\-gre\xspace}
\newcommand \qis {qua\-si in\-tè\-gres\xspace}

\newcommand \qnl {quasi-\nl}
\newcommand \qnls {quasi-\nls}

\newcommand \ralg {règle \agq}
\newcommand \ralgs {règles \agqs}

\newcommand \rav {racine virtuelle\xspace}
\newcommand \ravs {racines virtuelles\xspace}

\newcommand \rcc {\rcm con\-cret\xspace}
\newcommand \rca {\rcm abs\-trait\xspace}
\newcommand \rce {\rcc des é\-ga\-li\-tés\xspace}

\newcommand \rcm {recol\-lement\xspace}
\newcommand \rcms {recol\-lements\xspace}

\newcommand \rcv {recou\-vrement\xspace} 
\newcommand \rcvs {recou\-vrements\xspace}

\newcommand \rde {rela\-tion de dépen\-dance\xspace}
\newcommand \rdes {rela\-tions de dépen\-dance\xspace}

\newcommand \rdi {\rde inté\-grale\xspace}
\newcommand \rdis {\rdes inté\-grales\xspace}

\newcommand \rdl {\rde \lin}
\newcommand \rdls {\rdes \lins}

\newcommand \rdt {rési\-duel\-lement\xspace}

\newcommand \rdy {règle dyna\-mique\xspace}
\newcommand \rdys {règles dyna\-miques\xspace}

\newcommand \red {règle directe\xspace}
\newcommand \reds {règles directes\xspace}

\newcommand \rex {règle exis\-ten\-tielle simple\xspace}
\newcommand \rexs {règles exis\-ten\-tielles simples\xspace}

\newcommand \reX {\hyperref[defexistsimple]{règle exis\-ten\-tielle simple}\xspace}
\newcommand \reXs {\hyperref[defexistsimple]{règles exis\-ten\-tielles simples}\xspace}

\newcommand \rexri {règle exis\-ten\-tielle rigide\xspace}
\newcommand \rexris {règles exis\-ten\-tielles rigides\xspace}

\newcommand \rsim {règle de simplification\xspace}
\newcommand \rsims {règles de simplification\xspace}

\newcommand \rtl {réti\-culé\xspace}
\newcommand \rtls {réti\-culés\xspace}

\newcommand \rmq {\rcm de quotients\xspace} 
\newcommand \rvq {\rcv par quotients\xspace} 
\newcommand \rmqs {\rcms de quotients\xspace} %
\newcommand \rvqs {\rcvs par quotients\xspace} %

\newcommand \rpf {réduite-de-présen\-tation-finie\xspace}
\newcommand \rpfs {réduites-de-présen\-tation-finie\xspace}

\newcommand \rrl {relation de \recu \lin}
\newcommand \rrls {relations de \recu \lin}


\newcommand \sad {\salg dynamique\xspace}
\newcommand \sads {\salgs dynamiques\xspace}

\newcommand \sagq {semi\agq}
\newcommand \sagqs {semi\agqs}

\newcommand \sagc {\sagq continue\xspace}
\newcommand \sagcs {\sagqs continues\xspace}

\newcommand \salg {structure \agq}
\newcommand \salgs {structures \agqs}

\newcommand \scentrel {relation semi-implicative\xspace}
\newcommand \scentrels {relations semi-implicatives\xspace}

\newcommand \scf {schéma fini\-taire\xspace}
\newcommand \scfs {schémas fini\-taires\xspace}

\newcommand \scl {schéma \elr}
\newcommand \scls {schémas \elrs}

\newcommand \sdo {\sdr \orte}
\newcommand \sdos {\sdrs \ortes}

\newcommand \sdr {somme directe\xspace}
\newcommand \sdrs {sommes directes\xspace}

\newcommand \sdz {sans \dvz}

\newcommand \sfio {sys\-tème fondamental d'\idms ortho\-gonaux\xspace}
\newcommand \sfios {sys\-tèmes fondamentaux d'\idms ortho\-gonaux\xspace}

\newcommand \sgr {\sys \gtr}
\newcommand \sgrs {\syss \gtrs}

\newcommand \slgb {stricte\-ment \lgb}
\newcommand \slgbs {stricte\-ment \lgbs}

\newcommand \sli {\sys \lin}
\newcommand \slis {\syss \lins}

\newcommand \smq {symé\-trique\xspace}
\newcommand \smqs {symé\-triques\xspace}

\newcommand \spb {sépa\-rable\xspace}  
\newcommand \spbs {sépa\-rables\xspace}

\newcommand \spe {spéci\-fi\-cation\xspace}
\newcommand \spes {spéci\-fi\-cations\xspace}

\newcommand \spi {\spe incomplète\xspace}
\newcommand \spis {\spes incomplètes\xspace}

\newcommand \spl {sépa\-rable\xspace}  
\newcommand \spls {sépa\-rables\xspace}

\newcommand \spo {semipolynôme\xspace}
\newcommand \spos {semipolynômes\xspace}

\newcommand \spt{sépa\-ra\-bi\-lité\xspace}

\newcommand \srg {suite régu\-lière\xspace}
\newcommand \srgs {suites régu\-lières\xspace}

\newcommand \srl{suite r\'ecur\-rente \lin}
\newcommand \srls{suites r\'ecur\-rentes \lins}
\newcommand \Srls{Suites r\'ecur\-rentes \lins}

\newcommand \stf {strictement fini\xspace}
\newcommand \stfs {strictement finis\xspace}
\newcommand \stfe {strictement finie\xspace}
\newcommand \stfes {strictement finies\xspace}

\newcommand \stl {stablement libre\xspace}
\newcommand \stls {stablement libres\xspace}

\newcommand \stm {strictement\xspace}

\newcommand \str {\stm réticulé\xspace}
\newcommand \stre {\stm réticulée\xspace}
\newcommand \strs {\stm réticulés\xspace}
\newcommand \stres {\stm réticulées\xspace}

\newcommand \sul {supplé\-men\-taire\xspace}
\newcommand \suls {supplé\-men\-taires\xspace}

\newcommand \Sut {Support\xspace}
\newcommand \Suts {Supports\xspace}
\newcommand \sut {support\xspace}

\newcommand \syc {\sys de coordon\-nées\xspace}
\newcommand \sycs {\syss de coordon\-nées\xspace}

\newcommand \syp {\sys \poll}
\newcommand \syps {\syss \polls}

\newcommand \sys {sys\-tème\xspace}
\newcommand \syss {sys\-tèmes\xspace}

\newcommand \talg {théorie \agq}
\newcommand \talgs {théories \agqs}

\newcommand \tco {théorie cohé\-rente\xspace}
\newcommand \tcos {théories cohé\-rentes\xspace}

\newcommand \tdy {théorie dyna\-mique\xspace}
\newcommand \tdys {théories dyna\-miques\xspace}

\newcommand \tel {théorie exis\-ten\-tielle\xspace}
\newcommand \tels {théories exis\-ten\-tielles\xspace}

\newcommand \telri {théorie exis\-ten\-tielle rigide\xspace}
\newcommand \telris {théories exis\-ten\-tielles rigides\xspace}

\newcommand \tf {de type fini\xspace}

\newcommand \tfo {théorie formelle\xspace}
\newcommand \tfos {théorie formelles\xspace}

\newcommand \tgm {théorie \gmq}
\newcommand \tgms {théories \gmqs}

\newcommand \Tho {Théo\-rème\xspace}
\newcommand \Thos {Théo\-rèmes\xspace}
\newcommand \tho {théo\-rème\xspace}
\newcommand \thos {théo\-rèmes\xspace}

\newcommand \thoc {théo\-rème$\mathbf{^*}$~}

\newcommand \tpe {théorie \peq}
\newcommand \tpes {théories \peqs}

\newcommand \trdi {treil\-lis dis\-tri\-bu\-tif\xspace}
\newcommand \trdis {treil\-lis dis\-tri\-bu\-tifs\xspace}

\newcommand \trel {trans\-for\-mation \elr}
\newcommand \trels {trans\-for\-mations \elrs}

\newcommand \umd {unimo\-du\-laire\xspace}
\newcommand \umds {unimo\-du\-laires\xspace}

\newcommand \unt {uni\-taire\xspace}
\newcommand \unts {uni\-taires\xspace}

\newcommand \uvl {uni\-versel\xspace}
\newcommand \uvle {uni\-ver\-selle\xspace}
\newcommand \uvls {uni\-versels\xspace}
\newcommand \uvles {uni\-ver\-selles\xspace}


\newcommand \vfn {véri\-fi\-cation\xspace}
\newcommand \vfns {véri\-fi\-cations\xspace}

\newcommand \vmd {vec\-teur \umd}
\newcommand \vmds {vec\-teurs \umds}

\newcommand \zed {z\'{e}ro-di\-men\-sionnel\xspace}
\newcommand \zede {z\'{e}ro-di\-men\-sion\-nelle\xspace}
\newcommand \zeds {z\'{e}ro-di\-men\-sion\-nels\xspace}
\newcommand \zedes {z\'{e}ro-di\-men\-sion\-nelles\xspace}

\newcommand \zedr {\zed réduit\xspace}
\newcommand \zedrs {\zeds réduits\xspace}

\newcommand \zmt {\tho de Zariski-Grothen\-dieck\xspace}


\newcommand \cof {cons\-truc\-tif\xspace}
\newcommand \cofs {cons\-truc\-tifs\xspace}

\newcommand \cov {cons\-truc\-tive\xspace}
\newcommand \covs {cons\-truc\-tives\xspace}

\newcommand \coma {\maths\covs}
\newcommand \clama {\maths clas\-siques\xspace}

\renewcommand \cot {cons\-truc\-ti\-vement\xspace}

\newcommand \matn {mathé\-ma\-ticien\xspace}
\newcommand \matne {mathé\-ma\-ti\-cienne\xspace}
\newcommand \matns {mathé\-ma\-ticiens\xspace}
\newcommand \matnes {mathé\-ma\-ti\-ciennes\xspace}

\newcommand \maths {mathé\-ma\-tiques\xspace}
\newcommand \mathe {mathé\-ma\-tique\xspace}

\newcommand \prco {démons\-tration \cov}
\newcommand \prcos {démons\-trations \covs}



\theoremstyle{plain}
\newtheorem{ftheorem}{Théorème}[section]
\newtheorem{fthdef}[ftheorem]{Théorème et définition}
\newtheorem{fpstf}[ftheorem]{Positivstellensatz formel}
\newtheorem{fpst}[ftheorem]{Positivstellensatz}
\newtheorem{flemma}[ftheorem]{Lemme}
\newtheorem{fcorollary}[ftheorem]{Corolaire}
\newtheorem{fconjecture}[ftheorem]{Conjecture}
\newtheorem{fproposition}[ftheorem]{Proposition}
\newtheorem{fprpta}[ftheorem]{Propriétés attendues}
\newtheorem{fpropdef}[ftheorem]{Proposition et définition}
\newtheorem{ffact}[ftheorem]{Fait}
\newtheorem{fconvention}[ftheorem]{Convention}
\newtheorem{falgorithm}{Algorithme}

\newtheorem{ftheoremc}[ftheorem]{Th\'{e}or\`{e}me\etoz}
\newtheorem{flemmac}[ftheorem]{Lemme\etoz}
\newtheorem{fcorollaryc}[ftheorem]{Corolaire\etoz}
\newtheorem{fproprietec}[ftheorem]{Propri\'{e}t\'{e}\etoz}
\newtheorem{fpropositionc}[ftheorem]{Proposition\etoz}
\newtheorem{ffactc}[ftheorem]{Fait\etoz}

\newtheorem{fatheorem}{Théorème}[section]
\newtheorem{falemma}[fatheorem]{Lemme}
\newtheorem{facorollary}[fatheorem]{Corolaire}
\newtheorem{fapropriete}[fatheorem]{Propriété}
\newtheorem{faproposition}[fatheorem]{Proposition}
\newtheorem{fapropdef}[fatheorem]{Proposition et définition}
\newtheorem{fafact}[fatheorem]{Fait}

\theoremstyle{definition}
\newtheorem{frstr}[ftheorem]{Règles structurelles}
\newtheorem{frstra}[ftheorem]{Règles structurelles admissibles}
\newtheorem{fdefinition}[ftheorem]{Définition}
\newtheorem{fdfni}[ftheorem]{Définition informelle}
\newtheorem{fdefinitions}[ftheorem]{Définitions}
\newtheorem{fexample}[ftheorem]{Exemple}
\newtheorem{fexamples}[ftheorem]{Exemples}
\newtheorem{fnotation}[ftheorem]{Notation}
\newtheorem{fproblem}[ftheorem]{Problème}
\newtheorem{fquestion}[ftheorem]{Question}
\newtheorem{fquestions}[ftheorem]{Questions}

\newtheorem{fdefinitionc}[ftheorem]{Définition\etoz}
\newtheorem{fdefinota}[ftheorem]{Définition et notation} 
\newtheorem{faquestion}[fatheorem]{Question}
\newtheorem{fadefinition}[fatheorem]{Définition}

\theoremstyle{remark}
\newtheorem{fremark}[ftheorem]{Remarque}
\newtheorem{fremarks}[ftheorem]{Remarques}
\newtheorem{faremark}[fatheorem]{Remarque}
\newtheorem{faremarks}[fatheorem]{Remarques}
\newtheorem{fcomment}[ftheorem]{Commentaire}

\rdb
\label{beginfrench}

\FrenchFootnotes

\title{ Une variante de l'\algo de Berlekamp-Massey}

\author{
Nadia Ben Atti
(\thanks {~
Équipe de Mathématiques, CNRS UMR 6623, UFR des Sciences et
Techniques,
Université de Franche-Comté, 25 030 BESANCON cedex, FRANCE.
{\tt nadia.benatti@ensi.rnu.tn}}~),
Gema M. Diaz--Toca
(\thanks {~
Dpto. de Matematicas Aplicada.
Universidad de Murcia, Spain.
{\tt gemadiaz@um.es}}~)
Henri Lombardi
(\thanks{~
Équipe de Mathématiques, CNRS UMR 6623, UFR des Sciences et
Techniques,
Université de Franche-Comté, 25 030 BESANCON cedex, FRANCE,
{\tt henri.lombardi@univ-fcomte.fr},
partiellement financé par le  réseau européen RAAG
CT-2001-00271}~)
}

\date{2005, traduction 2022}

\emptythanks\setcounter{footnote}{0}
\setcounter{equation}{0}

\maketitle

\begin{abstract}
Nous donnons une variante, légèrement plus simple, de l'\algo de
Berlekamp-Massey. L'explication
de l'\algo est également plus facile. En outre cela autorise une 
approche \gui{dynamique} de cet \algo utile dans certaines situations
\end{abstract}
\bni MSC 2000: 68W30, 15A03

\mni Mots clés:  Algorithme de Berlekamp-Massey. \Srls.

\bni {\bf Note:} Cet article est la version française de l'article
\textsl{The Berlekamp-Massey Algorithm revisited} paru dans 
AAECC {\bf 17} 1 (2006), 75--82.
Received: 14 September 2005 / Published online: 9 March 2006

\setcounter{tocdepth}{4}

\startcontents[french]
\printcontents[french]{}{1}{}

\section{Introduction: l'\algo de Berlekamp-Massey usuel }

Soit $\mathbb{K}$ un corps arbitraire. Étant donné une \srl  notée $\mathcal{S}(x)=\sum_{i=0}^\infty a_i
x^i$, $a_i \in \mathbb{K}$, on désire calculer son \polmin, noté $P(x)$. Rappelons que si 
$P(x)=\sum_{i=0}^d p_i x^i$, alors
$P(x)$ est le \pol de degré le plus petit possible satisfaisant $
\sum_{i=0}^d p_i a_{j+i}=0, $ pour tous $j\in\N$.

Supposons que ce \polmin a son degré majoré par \(n\). Alors l'\algo de Berlekamp-Massey
utilise seulement les $2n$ premiers \coes de
$\mathcal{S}(x)$ par calculer ce \polmin. Ces \(2n\) \coes définissent le \pol $S=\sum_{i=0}^{2n-1}a_i\,x^i$.

Une vaste littérature est parue concernant cet \algo.
  L'\algo original a été créé en 1968 par Berlekamp (voir \cite{fBe}) pour décoder les codes de Bose-Chaudhuri-Hocquenghem (BCH). Une année plus tard, l'\algo a été simplifié dans \cite{fbr-ms}. La ressemblance entre cet \algo et l'\algo d'Euclide étendu 
a été soulignée dans plusieurs articles, par exemple \cite{fcheng}, \cite{fdorns}, \cite{fmilles}, \cite{fsugi} and
\cite{fwelch}. Des interprtétations plus récentes en termes de matrices de Hankel  et d'approximants de Padé
se trouvent dans \cite{fjon} et \cite{fpan}.

L'interprétation usuelle de l'\algo de Berlekamp-Massey est exprimée dans le pseudo code de l'Algorithme \ref{fBMA}.

\begin{falgor}[Algorithme de Berlekamp-Massey usuel]\label{fBMA}
\Entree Un entier $n\geq 1$. Une liste non nulle d'éléments du
corps $\K$, $[a_0,a_1,\ldots,a_{2n-1}]$: les $2n$ premiers termes
d'une \srl, sous l'hypothèse qu'elle admet un \polg de degré
$\leq n$. 
\Sortie Le \polgmin $P$ de la \srl. 
\Debut
\Varloc $R,R_0,R_1,V,V_0,V_1,Q$ : \pols en $X$  
\hst \# initialisation 
\hsu   $R_0:=X^{2n}$ ;
$R_1:=\sum_{i=0}^{2n-1}a_iX^i$ ;  $V_0=0$  ;
$V_1=1$ ; 
\hst \# boucle 
\hsu    \tantque{n \leq  \deg(R_1)} 
\hsd
$(Q,R) :=$ quotient et reste de la division de $R_0$ par $R_1$ ;
\hsd         $V := V_0-Q*V_1$ ; 
\hsd 
$V_0:=V_1$ ; $V_1:=V$ ; $R_0:=R_1$ ; $R_1:=R$ ; 
\hsu
\fintantque \hst \# sortie \hsu $d:=\sup(\deg(V_1),1+\deg(R_1)) $
; $P:=X^dV_1(1/X)$ ; Retourner $P=P/\cd(P)$. \fin
\end{falgor}

Bien que très simple cet \algo a toujours semblé un peu trop 
difficile
à justifier. Dans la littérature, il semble y avoir une petite confusion concernant la \dfn du \polmin d'une \srl.
Ici, nous introduisons dans la section 2 une légère modification de l'\algo qui le rend plus naturel et plus facile à justifier. Nous n'avons pas trouvé trace de cette variante dans la littérature parue avant notre première soumission à Mathematics of Computation en 2004. Cette variante est apparue dans le livre de Victor Shoup  (\cite{fshoup}), 
sans aucune référence à un article antérieur.

\section{De bonnes raisons pour modifier l'\algo usuel}

La variante que nous introduisons est basée sur l'idée suivante.
Puisqu'à la fin de l'\algo, le \pol $V$ doit \^{e}tre renversé 
selon
une prodédure difficile à éxpliquer (pourquoi doit-on prendre 
le \pol
réciproque en degré $d=\sup(\deg(V_1),1+\deg(R_1))$?), le mieux ne
serait-il pas de traiter directement la \srl \gui{dans le bon sens} 
(elle
a elle-m\^{e}me été renversée au départ lorsqu'on a affecté
$R_1:=\sum_{i=0}^{2n-1}a_iX^i$)?

Naturellement l'appréciation selon laquelle la \srl a été 
renversée
au départ peut sembler subjective. Elle est en fait renforcée par 
la
remarque suivante. Si le \polgmin est de degré $d$ nettement plus 
petit
que $n$ les calculs dans l'\algo ne devraient pas \^{e}tre 
sensiblement
différents lorsqu'on travaille avec les $2d$ premiers termes de la 
suite
ou   lorsqu'on travaille avec les $2n$ premiers termes. Or le 
renversement
effectué au début de l'\algo change complètement le calcul qui 
est
fait. Tandis qu'en l'absence de renversement, avec notre variante, le 
calcul
sur la suite courte peut facilement \^{e}tre regardé comme le 
calcul sur
la suite longue, tronqué de manière convenable.

Une autre confirmation est le caractère plus simple et plus facile 
à
justifier dans l'affectation finale.

\begin{falgor}[Algorithme de Berlekamp-Massey, variante]\label{fBMA2}
\Entree Un entier $n\geq 1$. Une liste non nulle d'éléments du
corps $\K$, $[a_0,a_1,\ldots,a_{2n-1}]$: les $2n$ premiers termes
d'une \srl, sous l'hypothèse qu'elle admet un \polg de degré
$\leq n$. 
\Sortie Le \polgmin $P$ de la \srl. 
\Debut 
\Varloc $R,R_0,R_1,V,V_0,V_1,Q$ : \pols en $X$ ; $m=2n-1$ : 
entier.
\hst \# initialisation 
\hsu $m:=2n-1$ ;
$R_0:=X^{2n}$ ; $R_1:=\sum_{i=0}^{m}a_{m-i}X^i$ ; 
$V_0=0$ ;  $V_1=1$ ; 
\hst \# boucle 
\hsu    \tantque{n
\leq  \deg(R_1)} 
\hsd        $(Q,R) :=$ quotient et reste de la
division de $R_0$ par $R_1$ ; 
\hsd    $V :=V_0-Q*V_1$ ; 
\hsd    $V_0:=V_1$ ;
$V_1:=V$ ; $R_0:=R_1$ ; $R_1:=R$ ; 
\hsu \fintantque \hst \# sortie
\hsu Retourner $P:=V_1/\cd(V_1)$. \fin
\end{falgor}

Nous allons maintenant prouver la correction de cet \algo.

Si \,$\s{a}=(a_n)_{n\in \N}$\, est une suite arbitraire
et si \,$i,$ $r,$ $p\in\N$\,  nous noterons
\,$\rH^{\s{a}}_{i,r,p}$\, la matrice de Hankel suivante,
qui possède \,$r$\, lignes et \,$p$\, colonnes:
$$\rH^{\s{a}}_{i,r,p}=
\left[ \begin{array}{ccccc}
a_i & a_{i+1} & a_{i+2}   &\ldots & a_{i+p-1}  \\
a_{i+1} & a_{i+2} &       &       & a_{i+p}  \\
a_{i+2} &     &       &       &     \\
\vdots &  &       &       & \vdots   \\
a_{i+r-1}&a_{i+r}&\ldots&\ldots & a_{i+r+p-2}
\end{array}\right]\,.
$$
et nous noterons $\rP^{\s{a}}(X)$ le \polgmin de $\s a$.

La proposition suivante est classique (voir par exemple 
\cite{fjon}).
\begin{fproposition} \label{fpropSRL}
Si \,$\s{a}$\, est une \srl qui admet un
\pol générateur de degré $\leq n$,
 alors le degré \,$d\leq n$\, de son \pol générateur minimal
\,$\rP^{\s{a}}$\, est égal au rang de la matrice de Hankel
$$\rH^{\s{a}}_{0,n,n}=
\left[ \begin{array}{cccccc}
a_0 & a_{1} & a_{2}   &\cdots & a_{n-1}  \\
a_{1} & a_{2} &       & \adots  & a_{n}  \\
a_{2} &     & \adots & \adots & \vdots    \\
\vdots &\adots   & \adots  &       & \vdots   \\
a_{n-1}&a_{n}&\cdots&\cdots & a_{2n-2}
\end{array}\right] \;
.$$
Les \coes de
\,$\rP^{\s{a}}(X)=X^d-\sum_{i=0}^{d-1}g_iX^i\in\K[X]\,$
sont l'unique solution de l'équation
$$\rH^{\s{a}}_{0,d,d}\,G=\rH^{\s{a}}_{1,d,d}
$$
\cad encore l'unique solution du \sli
\begin{equation} \label{fEqSRL3}
\left[ \begin{array}{ccccc}
a_0 & a_{1} & a_{2}   &\cdots & a_{d-1}  \\
a_{1} & a_{2} &       & \adots  & a_{d}  \\
a_{2} &     & \adots & \adots & \vdots    \\
\vdots &\adots   & \adots  &       & \vdots   \\
a_{d-1}&a_{d}&\cdots&\cdots & a_{2d-2}
\end{array}\right] \;
\left[ \begin{array}{c}
g_0 \\ g_1 \\ g_2 \\  \vdots   \\  g_{d-1}
\end{array}\right]
=
\left[ \begin{array}{c}
a_d \\  a_{d+1} \\ a_{d+2} \\  \vdots \\ a_{2d-1}
\end{array}\right] \,.
\end{equation}
\end{fproposition}
\begin{fcorollary} \label{corfpropSRL}
On reprend les notations de la proposition \ref{fpropSRL}. Pour qu'un vecteur $Y=(p_0,\ldots,p_n)$ soit solution de 
l'équation
$$
\rH^{\s{a}}_{0,n,n+1}\,Y=0
$$
\cad
\begin{equation} \label{fEqSRL5}
\left[ \begin{array}{cccccc}
a_0 & a_{1} & a_{2}   &\cdots & a_{n-1}&a_{n}  \\
a_{1} & a_{2} &       & \adots  & a_{n}& a_{n+1}  \\
a_{2} &     & \adots & \adots & \vdots& \vdots    \\
\vdots &\adots   & \adots  &       & \vdots& \vdots   \\
a_{n-1}&a_{n}&\cdots&\cdots & a_{2n-2}
& a_{2n-1}
\end{array}\right] \;
\left[ \begin{array}{c}
p_0 \\ p_1 \\ p_2 \\  \vdots   \\  p_{n-1}\\  p_{n}
\end{array}\right]
= 0
\end{equation}
il faut et il suffit que le \pol $P(X)=\sum_{i=0}^{n}p_iX^i\in\K[X]$ 
soit multiple de
$\rP^{\s{a}}(X)$.

\end{fcorollary}

Par ailleurs, nous faisons la constatation suivante.
\begin{ffact}
\label{ffactEvident}
En posant $m:=2n-1$ et   $S:=\sum_{i=0}^{m}a_{m-i}X^i$, l'équation
(\ref{fEqSRL5}) est équivalente à l'affirmation suivante.
Le \pol $P$ est de degré $\leq n$ et on a:
$$
\exists R\in\K[X]\;\mathrm{tel\; que}\quad \deg(R)<n\quad 
\mathrm{et}\quad
P(X)S(X)\equiv R(X)\quad \mod X^{2n}.
$$
\cad encore
\begin{equation} \label{fEqSRL6}
\exists R,U\in\K[X]\;\mathrm{tels\; que}\quad \deg(R)<n\quad
\mathrm{et}\quad P(X)S(X)+U(X)X^{2n}= R(X)
\end{equation}
\end{ffact}
Le problème de trouver le \polgmin de $\s{a}$ est donc ramené au
problème de réaliser (\ref{fEqSRL6}) avec le degré de $P$ 
minimum.
Or il est bien connu que:
\begin{itemize}
\item l'\algo d'Euclide étendu démarrant avec $R_0=X^{2n}$ et 
$R_1=S$
réalise une équation telle que (\ref{fEqSRL6}) avec $R=R_k$ 
lorsque le
premier reste $R_k$ de degré $<n$ est atteint,
\item quand une équation telle que (\ref{fEqSRL6}) est réalisée 
par
l'\algo d'Euclide étendu, il n'est pas possible d'obtenir une 
équation
du m\^{e}me type $P'(X)S(X)+U'(X)X^{2n}= R'(X)$ avec $\deg(R')< 
\deg(R_{k-
1})$ et $\deg(P')< \deg(P)$.
\end{itemize}

Ceci prouve que notre algorithme est correct.

\section{Une variante paresseuse}

Outre sa simplicité, notre variation sur l'\algo de 
Berlekamp-Massey 
admet une variante paresseuse, qui peut \^etre utile dans certains cas.

Par exemple considérons l'algèbre $B$ de décomposition universelle d'un
\pol $f(X)$ (séparable) sur un corps $\K$, ou plus 
généralement dans un quotient $A$ de $B$. $A$ est une algèbre
zéro-dimensionnelle dont on conna\^{\i}t une présentation:
$A\simeq\K[X_1,\ldots ,X_n]/\gen{f_1,\ldots ,f_s}$, o\`u $f_1,\ldots 
,f_s$ est une base de Gr\"obner.
La dimension $m$ de $A$ sur $\K$ est trop grande pour qu'on envisage 
de manipuler des matrices dans $\K^{m\times m}$. 
On est intéressé pour calculer le \polmin d'un \elt $x$ de $A$, 
ou à défaut un facteur de ce \polmin. 
On utilise pour cela l'\algo de Wiedemann. On calcule donc les termes 
successifs d'une \srl: $a_n=\phi(x^n)$  o\`u $\phi$ est une forme 
linéaire sur $A$. Le calcul de $x^n$ étant assez co\^uteux (on 
doit à chaque fois calculer la nouvelle forme normale), et le 
\polmin étant a priori de degré nettement plus petit que la 
dimension de $A$ (vu le choix judicieux de $x$), on a intér\^et à 
tester
régulièrement si le \polmin de $x$ n'est pas atteint ``avant 
terme". \ldots

\newpage

 \begin{falgorH}[L'\algo de Berlekamp-Massey paresseux (dans un contexte particulier)]\label{fBMAP}
\Entree  $ m \in \N $, $C \in \K^n$, $G$: base de Gr\"obner, $a \in
A$. Le \polmin est de degré  $\leq m$.
\Sortie Le \polmin $P$ de $a$
\Debut \Varloc $l,i$: entiers,
$R,R_{-1},R_0,R_1,V,V_{-1},V_0,V_1,U,U_{-1},U_0,U_1, S_0,S_1,Q$ :
\pols en $x$, $L,W$: listes dans \(A\), \(val\): \elt de \(A\);
\hst \# initialisation
  \hsu
$ l =\lfloor m/4 \rfloor$;
\hsu $L:=[1,a]$;
$W:=[1,\mathrm{Value}(a,C)]$;
\hsu $S_0:=x^{2l}$ ; $S_1 =W[1]\,x^{2l-1}+W[2]\,x^{2l-2}$;
  \hst \# boucle
\hsu \por{i}{3}{2l}\hsd
         $L[i]:=\mathrm{normalf}(L[i-1]a ,G);$
           $V[i]:=\mathrm{Value}( L[i],C);$
         $S_1 = S_1 + V[i] x^{2l-i};$
\hsu
\finpour 
\hsu
$R_0:=S_0; R_1:=S_1$; 
$V_0=0$ ; $V_1=1$ ; $U_0=1$ ; $V_1=0$;
   \hst \# boucle \hsu
\tantque{l \leq \deg(R_1)} 
\hsd $(Q,R) :=$ quotient et reste
de $R_0$ divisé par $R_1$ ; 
\hsd $V := V_0-Q V_1$; $U := U_0-Q
U_1$ ;  $U_{-1}:= U_0$; $V_{-1}:= V_0$; 
\hsd $V_0:=V_1$ ; $V_1:=V$ ; $U_0:=U_1$ ; $U_1:=U$; $R_0:=R_1$ ; $R_1:=R$ ;
\hsu
\fintantque
\hsu
\(val:=\mathrm{Subs}(x=a,V_1)\);
   \hst \# boucle 
\hsu  \tantque{ val \neq 0 } 
  \hsd $l:=l+1;$
  \hsd  \# boucle
\hsd \por{i}{2l-1}{2l} 
\hst
         $L[i]:=\mathrm{normalf}(L[i-1]a ,G)$; \hst
         $W[i]:=\mathrm{Value}( L[i],C)$;
\hsd \finpour 
\hsd
  $S_0=x^2 S_0$; $S_1 = x^2S_1 + W[2l-1] x+ W[2l];$
\hsd
$R_0:=U_{-1} S_0+V_{-1} S_1$; $R_1:=U_0 S_0+V_0 S_1$ ; 
  \hsd
$U_1:=U_0; V_1:=V_0; U_0:=U_{-1}; V_0:=V_{-1};$
  \hsd \# boucle \hsd
\tantque{l \leq \deg(R_1)} 
\hst $(Q,R) :=$ quotient et reste
de $R_0$ divisé par $R_1$ ; 
\hst $V := V_0-Q V_1$; $U := U_0-Q
U_1$; $U_{-1}:= U_0$; $V_{-1}:= V_0$; 
\hst $V_0:=V_1$ ; $V_1:=V$
;$U_0:=U_1$ ; $U_1:=U$;
  $R_0:=R_1$ ; $R_1:=R$ ;
\hsd
\fintantque
  \hsd
  \(val:=\mathrm{Subs}(x=a,V_1)\)
\hsu \fintantque
\hst \# sortie \hsu  Retourner $P:=V_1/
\mathrm{leadcoeff}(P)$. \fin
\end{falgorH}
  
\addcontentsline{toc}{section}{Références}

\rdb

\stopcontents[french]

\end{document}